\documentclass[a4paper]{amsart}
\usepackage[foot]{amsaddr}

\usepackage{amsmath}
\usepackage{amssymb}
\usepackage{amsthm}
\usepackage{bm}
\usepackage{booktabs}
\usepackage{fullpage}
\usepackage{hyperref}
\usepackage[english]{isodate}
\usepackage{color}
 
\theoremstyle{plain}

\newtheorem{lemma}{Lemma}
\newtheorem{theorem}{Theorem}

\theoremstyle{definition}
\newtheorem{definition}{Definition}
\newtheorem{observation}{Observation}

\newtheorem{remark}{Remark}

\def\Z{{\mathbb Z}}

\renewcommand{\leq}{\leqslant}
\renewcommand{\geq}{\geqslant}

\usepackage{multirow}
\usepackage{rotating}
\usepackage{array}
\newcolumntype{L}{>{\centering\arraybackslash}m{10cm}}
 \usepackage{graphicx} 
\newsavebox{\mybox}
\newenvironment{mytabularwrap}{\begin{lrbox}{\mybox}}
	{\end{lrbox}%
	\setbox0\hbox{\usebox\mybox}%
	\ifdim\wd0<\textwidth
	\usebox\mybox%
	\else
	\resizebox{\textwidth}{!}{\usebox\mybox}%
	\fi
}
\title{The metric dimension of the circulant graph $C(n,\pm\{1,2,3,4\})$}

\author{Cyriac Grigorious}
\author{Thomas Kalinowski}
\author{Joe Ryan}
\author{Sudeep Stephen}

\address{School of Mathematical and Physical Sciences, University of Newcastle, 2308 NSW, Australia}

\date{\today}

\begin{document}

\begin{abstract}
  Let $G=(V,E)$ be a connected graph and let $d(u,v)$ denote the distance between vertices
  $u,v \in V$. A metric basis for $G$ is a set $B\subseteq V$ of minimum cardinality such that no
  two vertices of $G$ have the same distances to all points of $B$. The cardinality of a metric basis of $G$
  is called the metric dimension of $G$, denoted by $\dim(G)$. In this paper we determine the 
  metric dimension of the circulant graphs $C(n,\pm\{1,2,3,4\})$ for all values of $n$.
\end{abstract}

\maketitle

\section{Introduction}

Let $G(V,E)$ be a simple connected and undirected graph. For $u,v \in V$, let $d(u,v)$ denote the
distance between $u$ and $v$. A vertex $x \in V$ is said to resolve or distinguish two vertices $u$
and $v$ if $d(x, u) \neq d(x, v)$. A set $X \subseteq V$ is said to be a \textit{resolving set} for
$G$, if every pair of vertices of $G$ is distinguished by some element of $X$. A minimum resolving
set is called a \textit{metric basis.} The cardinality of a metric basis is called the \emph{metric
  dimension} of $G$, denoted by $\dim(G)$. For an ordered set
$X =\{x_1, x_2, \dotsc, x_k\} \subseteq V$, we refer to the $k$-vector (ordered $k$-tuple)
$r(v|X) = (d(v,x_1), d(v,x_2), \dotsc, d(v,x_k))$ as the \emph{representation} of $v$ with respect
to $X$. Thus we can have another equivalent definition. The set $X$ is called a \emph{metric basis}
if distinct vertices of $G$ have distinct representation with respect to $X$. The metric dimension
problem has been studied in several papers
\cite{ChartrandErohJohnsonEtAl2000,HararyMelter1976,KhullerRaghavachariRosenfeld1996,Slater75},
where it is also referred to as the \emph{location number}. Note that metric basis, \emph{minimum
  locating set} and \emph{reference set} are different names used by different authors to describe
the same concept. In this paper we use the terms metric basis and metric dimension.

The problem of finding the metric dimension of a graph was studied by Harary and Melter
\cite{HararyMelter1976}. Slater described the usefulness of this idea in long range aids to
navigation~\cite{Slater75}. Melter and Tomescu~\cite{MelterTomescu1984} studied the metric dimension
problem for grid graphs. The metric dimension problem has been studied also for trees and
multi-dimensional grids by Khulller, Raghavachari and
Rosenfield~\cite{KhullerRaghavachariRosenfeld1996}. They also described the application of the concept of
metric dimension in robot navigation and in~\cite{ChartrandErohJohnsonEtAl2000} Chartrand, Eroh, Johnson and Oellermann presented
an application in drug discovery, where it is to be determined whether the features of a compound
are responsible for its pharmacological activity.

Cayley graphs on the cyclic group $\Z_n$ are called \textit{circulant graphs}. We use the special
notation $C(n,S)$ for a circulant graph on $\Z_n$ with connection set $S$.  In this paper we focus
on connection sets of the form $\{1,\dotsc,t\}$. Let $C(n, \pm \{1,2, \dotsc, t\})$ for
$1 \leq t \leq \lfloor n/2\rfloor$ and $n\geq3$, denote the graph with vertex set $V=\{0,1, \dotsc, n-1\}$
and edge set $E=\{(i,j):|j- i|\equiv s\pmod n,s\in\{1,2, \dotsc, t\}\}$. Note that $C(n, \pm \{1\})$
is the cycle $C_n$, and $C(n, \pm\{1,2, \dotsc, \lfloor n/2 \rfloor \})$ is the complete graph
$K_n$.  The distance between two vertices $i$ and $j$ in a circulant graph
$C(n,\pm\{1,2,\dotsc,t\})$ such that $0 \leq i <j<n$
is,
\[d(i,j)=\begin{cases}
\left\lceil(j-i)/t\right\rceil& \text{for } j-i\in\{0,1,\dotsc\lfloor n/2\rfloor\}, \\
\left\lceil(n-(j-i))/t\right\rceil& \text{for }j-i\in\{\lceil n/2\rceil,\dotsc, n-1\}.
\end{cases}\]
Extending the results of Imran, Baig, Bokhary and
Javaid~\cite{ImranBaigBokharyEtAl2012}, in~\cite{BorchertGosselin2014} Borchert and Gosselin showed
that $\dim(C(n,\pm\{1,2\}) = 4$ if $n \equiv 1 \pmod {4}$ and $\dim(C(n,\pm\{1,2\}) = 3$
otherwise. They also solved the case $t=3$ by proving
\[\dim(C(n,\pm \{1,2,3\}))=
  \begin{cases}
      5 & \text{for } n\equiv 1\pmod 6, \\
      4 & \text{otherwise.} 
  \end{cases}
\]
For general $t$, the following bounds were obtained by Vetrik in~\cite{Vetrik2016}.
\begin{itemize}
\item For $t\geq 2$ and $n \geq t^2+1$, we have $\dim C(n, \pm \{1, 2, \dotsc, t\}) \geq t$.
\item For $t\geq 2$ and $n=2kt+r$ with $k \geq 0$ and $t+2 \leq r \leq 2t+1$ we have
  $\dim C(n,\pm\{1,2, \dotsc, t\}) \geq t+1$.
\item For even $t$ and $n=2kt+t+2p$, we have $\dim C(n,\pm\{1,2, \dotsc, t\})\leq t+p$.
\end{itemize}
Actually, the condition $n\geq t^2+1$ in the first point above can be relaxed: 
\begin{itemize}
\item For $t\geq 2$ and $n \geq 2t+1$, we have $\dim C(n, \pm \{1, 2, \dotsc, t\}) \geq t$.
\end{itemize}
We will prove this in Section~\ref{sec:auxiliary}. In~\cite{GrigoriousManuelMillerEtAl2014}, we gave an
upper bound as follows:
\begin{itemize}
\item For $t\geq 2$ and $n=2kt+r$ with $k\geq 1$ and $2\leq r \leq t+2$ we have
  $\dim(C(n,\pm\{1,2,\dotsc,t\}))\leq t+1$.
\end{itemize}
Combining all these bounds we obtain the following for $t=4$:
\begin{itemize}
\item $\dim C(n, \pm \{1, 2, 3, 4\})=4$ for $n\equiv 4\pmod 8$,
\item $\dim C(n, \pm \{1, 2, 3, 4\})\in\{4,5\}$ for $n\equiv 2,3,\text{ or }5\pmod 8$,
\item $\dim C(n, \pm \{1, 2, 3, 4\})=5$ for $n\equiv 6\pmod 8$,
\item $\dim C(n, \pm \{1, 2, 3, 4\})\in\{5,6\}$ for $n\equiv 0\pmod 8$,
\item $\dim C(n, \pm \{1, 2, 3, 4\})\geq 5$ for $n\equiv \pm 1\pmod 8$.  
\end{itemize}
In this paper, we determine $\dim C(n,\pm\{1,2,3,4\})$ for all values of
$n$ as follows.
\begin{theorem}\label{mainthm}
Let $G=C(n, \pm \{1, 2, 3, 4\})$, $n\geq 6$, $n\not\in\{11,19\}$. Then 
\[\dim(G)=\begin{cases}
 4 & \text{for }\ n\equiv 4 \pmod 8, \\  
 5 & \text{for }\ n\equiv \pm 2\text{ or }\pm 3 \pmod 8, \\
 6 & \text{for }\ n\equiv \pm 1\text{ or }0 \pmod 8.
\end{cases}\]
For $n\in\{5,11,19\}$, we have $\dim(G)=4$.
\end{theorem}
\begin{remark}
  The cases $n\in\{5,\,11,\,19\}$ have to be treated separately in Theorem~\ref{mainthm}, because
  $\{0,\,1,\,2,\,3\}$, $\{0,\,2,\,3,\,10\}$ and $\{0,\,2,\,7,\,19\}$, respectively, are resolving
  sets witnessing $\dim(G)=4$ in these cases.
\end{remark}

\begin{remark}
  While we were working on the revision of the present paper we became aware of the recent work by Chau and
  Gosselin~\cite{ChauGosselin} who independently obtained some of our results. In particular, they
  prove the cases $n\equiv 1\pmod 8$ and $n\equiv 3\pmod 8$ of Theorem~\ref{mainthm}.
\end{remark}

The paper is structured as follows. In Section~\ref{sec:notation} we introduce definitions and
notation. In Section~\ref{sec:auxiliary} we prove some lemmas that are needed in the proof of
Theorem~\ref{mainthm} which is contained in Sections~\ref{sec:proof} and~\ref{sec:upper}, where
Section~\ref{sec:proof} is devoted to establishing the lower bounds, while Section~\ref{sec:upper}
contains the upper bounds which do not already follow from the results mentioned above, i.e., the
cases $n=8k+r$ with $r \in\{7,9\}$.

\section{Definitions and notations}\label{sec:notation}
Throughout this paper we consider the graph $G=C(n,\pm\{1,\dotsc,t\})$. We prove a few results for
general $t$, but mostly we are concerned with the case $t=4$. For $t=4$ we write the order of the
graph as $n=8k+r$ with $k\geq 1$ and $r \in \{2,3, \dotsc, 9\}$ (the cases $n\leq 9$ can be done by
a brute force computer search). The graph $G$ has diameter $k+1$, and for every vertex
$v$, the set
\[D_v=\{v+4k+j\,:\ j=1,2, \dotsc, r-1\}\]
of vertices at diameter distance from $v$ has size $\lvert D_v\rvert=r-1$.

We generally consider the vertex set of a circulant graph of order $n$ as $V=\{0,1,2, \dotsc,n-1\}$,
and whenever we refer to a vertex such as $i+j$ or $-4m$ this has to be interpreted modulo $n$ in
the obvious way.
\begin{definition}
  For any $S\subseteq V$ we define an equivalence relation $\sim_S$ by
  $u\sim_S v \iff r(u|S)=r(v|S)$.
\end{definition}
A set $S\subseteq V$ is a resolving set if and only if all equivalence classes of $\sim_S$ are
singletons. So the equivalence classes can serve as a measure for how far away the set $S$ is from
being a resolving set. If we want to extend a set $S$ to become a resolving set, we have to add a set
$X$ of vertices that resolve all the non-trivial equivalence classes. In our proof it will be
convenient to consider subsets of equivalence classes, which we call $S$-blocks. This is made more
precise in the following definition. 
\begin{definition}
  Let $A\subset V$ and $S\subset V$. We call the set $A$ an \emph{$S$-block}, if all vertices of $A$
  are at equal distance from every vertex of $S$, or equivalently,
  $r(a\vert S)=r(b\vert S)$ for all $a,b\in A$. Slightly abusing notation we denote this common
  representation vector by $r(A\vert S)$.
\end{definition}
\begin{definition}
  Let $A\subset V$ and $X \subset V$ be such that for all $a,b \in A$ we have $r(a|X)\neq
  r(b|X)$. Then we say $X$ \emph{resolves} $A$.
\end{definition}
For given $S$, if we want to find a vertex set $X$ such that $S\cup X$ is a resolving set, then $X$
has to resolve all the $S$-blocks simultaneously. In the following definition we introduce notation
for a collection of $S$-blocks.
\begin{definition}
  Let $S\subseteq V$. An $\ell$-tuple $A=(A_1, A_2, \dotsc, A_\ell) $ of $S$-blocks is called an
  \emph{$S$-cluster} if the sets $A_i$ are subsets of distinct $\sim_S$-classes. We say that a set
  $X\subseteq V-S$ resolves the $S$-cluster $A$, if $r(a|X) \neq r(b|X)$ for all
  $a\neq b \in A_j$, $1\leq j \leq\ell$.
\end{definition}
Note that the set $\{i,i+1\}$ can be resolved by a vertex $x$ with $ d(x,i) < d(x,i+1)$ which implies
$x \in \{i-4j\,:\,0 \leq j \leq k\}$, or by a vertex $x$ with $d(x,i) > d(x,i+1)$ which implies
$x \in \{i+1+4j\,:\,0 \leq j \leq k\}$. So
\[R_{i}=\{i-4j\ :\ 0 \leq j \leq k\}\cup\{i+1+4j\ :\ 0 \leq j \leq k\}\]
is the set of vertices that resolve the set $\{i,i+1\}$. In particular, every metric basis must
contain at least one element from each of the sets $R_i$.

\section{Auxiliary results}\label{sec:auxiliary}
In order to prove lower bounds in Theorem~\ref{mainthm}, we need to go into rather tedious case
discussions. The basic idea is to show that a resolving  set $B$ of size $k$ cannot exist by looking
at all possible ways of starting with a set $S$, $\lvert S\rvert=k_0<k$, and exhibit an $S$-block or
$S$-cluster whose resolution requires more than $k-k_0$ vertices. For this purpose we need many
statements of the form ``If $A$ is an $S$-cluster of the form \ldots, then every set resolving $A$
has at least \ldots elements.'' The statements used in the proof of Theorem~\ref{mainthm} are the
lemmas proved below. Sometimes we need supporting claims in the proofs of the lemmas, and we call them
\emph{observations}.

We start with a result which is valid for any $t$, and use it to prove two general lower bounds. 
\begin{lemma}\label{k-2resolving}
  Let $G=C(n, \pm \{1,2, \dotsc, t\})$, and let $A\subseteq\{i, i+1, \dotsc, i+t\}$ with
  $\lvert A\rvert=\ell$, for $2 \leq \ell \leq t+1$. If $X$ resolves $A$, then $\lvert X\rvert\geq \ell-1$.
\end{lemma}
\begin{proof}
We proceed by induction on $\ell$. For $\ell=2$ the statement is true as we need at least one vertex
to resolve a set of size at least $2$. For $\ell\geq 3$, consider the elements of $A$ in the order
$i+a_1,i+a_2,\dotsc,i+a_\ell$ with $0\leq a_1<a_2<\dotsb<a_\ell\leq t$. By assumption there is an
element $x\in X$ with $d(x,i+a_1)\neq d(x,i+a_2)$. If $x=i+a_2$ then $d(x,i+a_s)=1$ for all
$s\in\{1,3,4,\dotsc,\ell\}$. Therefore $X\setminus\{x\}$ resolves $A\setminus\{i+a_2\}$. If $x\neq
i+a_2$, then $d(x,i+a_s)=d(x,i+a_2)$ for all $s\in\{2,\dotsc,\ell\}$. Therefore $X\setminus\{x\}$
resolves $A\setminus\{i+a_1\}$. In both cases $X\setminus\{x\}$ resolves a subset of
$\{i,\dotsc,i+t\}$ of size $\ell-1$, and by induction $\lvert X\setminus\{x\}\rvert\geq \ell-2$, hence $\lvert X\rvert\geq\ell-1$.
\end{proof}
The following theorem strengthens the lower bound given by Vetrik in~\cite{Vetrik2016} by extending the range of
$n$ for which the lower bound is valid. Note that if we have a circulant graph
$C(n,\{1,2,\dotsc,t\})$ with $n<2t+2$, then that circulant graph is a complete graph on $n$ vertices
and its metric dimension is $n-1$.
\begin{theorem}\label{general_t}
  Let $G=C(n,\pm\{1,2,\dotsc, t\})$, with $n\geq2t+2$. Then $\dim(G) \geq t$.
\end{theorem}
\begin{proof}
  Suppose $B$ is a resolving set of $G$ with $\lvert B\rvert< t$. Without loss of generality,
  $S=\{0\}\subseteq B$. The set $A=\{1,2,\dotsc, t\}$ is an $S$-block, and by
  Lemma~\ref{k-2resolving}, $\lvert B-S\rvert\geq t-1$, which is the required contradiction.
\end{proof}
Lemma~\ref{k-2resolving} can also be used to give a short alternative proof of the following result
which was originally proved by Vetrik in~\cite{Vetrik2016}.
\begin{theorem}[\cite{Vetrik2016}]
  Let $n=2kt+r$ where $t\geq 2$, $k\geq 1$ and $r \in \{t+2, t+3, \dotsc,
  2t+1\}$. Then
  \[\dim(C(n,\{1,2,\dotsc,t\}))\geq t+1.\]
\end{theorem}
\begin{proof}
  Suppose $B$ is a resolving set of $G$ with $\lvert B\rvert<t+1$. Without loss of generality,
  $S=\{0\}\subseteq B$. For $D_0=\{tk+\ell : \ell = 1,2,\dotsc,r-1\}$, we have
  $\lvert D_0\rvert=r-1\geq t+1$, and $d(0,u)=k+1$ for all $u \in D_0$. Hence, $D_0$ is an
  $S$-block, and by Lemma~\ref{k-2resolving}, $\lvert B-S\rvert\geq t$, which is the required
  contradiction.
\end{proof}
For the rest of this section, $t=4$.
\begin{lemma}\label{lem:min_dist_789}
  Let $n=8k+r$ with $r\in\{7,8,9\}$, and let $G=C(n,\pm\{1,2,3,4\})$. Suppose $B$ is a metric basis
  for $G$ with $\lvert B\rvert=5$. Then $\lvert i-j \rvert\geq r-5$ for all $i\neq j\in B$.
\end{lemma}
\begin{proof}
  If the statement is wrong, then without loss of generality, $S=\{0,i\}\subseteq B$ for some
  $i\in\{1,\dotsc,r-6\}$. The set $A=\{4k+i+\ell\,:\,\ell=1,\dotsc,5\}$ is an $S$-block with
  $r(A\vert S)=(k+1,k+1)$ and, since $B$ is a metric basis, $B-S$ resolves $A$, and by
  Lemma~\ref{k-2resolving} this implies $\lvert B-S\rvert \geq
  4$. Hence, $\lvert B\rvert\geq 6$, which is the required contradiction.
\end{proof}
\begin{observation}\label{observation_0123}
  Let $G=C(n,\pm\{1,2,3,4\})$ with $n=8k+r$, $r\in \{2,\dotsc,9\}-\{3\}$ and
  let
  \[A=(\{a,a+1\},\{a+2,a+3\})\]
  be an $S$-cluster for some $S$. Then for every resolving set $X$
  of $A$, $|X|\geq 2$.
\end{observation}
\begin{proof}
  Without loss of generality $a=0$. Suppose $x\in V$ resolves $A$, that is,
  \[x\in R_0\cap R_2=\{-4m,1+4m\,:\,0\leq m\leq k\}\cap\{2-4m,3+4m\,:\,0\leq m\leq k\}=\emptyset,\]
  which is the required contradiction.
\end{proof}
\begin{lemma}\label{2-4-3lemma_1}
Let $n=8k+9$, $G=C(n,\pm\{1,2,3,4\})$, and let
\[A=\left(\{a,a\pm1\},\,\{a\pm j\pm4\ell\,:\,j\in\{2,3,4\}\}\right)\] 
be an $S$-cluster for some $S$, where $0 \leq \ell \leq 2k+1$. 
If $X \subseteq V$ resolves $A$, then $\lvert X\rvert\geq 3$.
\end{lemma}
\begin{proof}
  By symmetry, it is sufficient to prove the statement for
  $A=(\{0,1\},\{j+4\ell:j\in\{2,3,4\}\})$. Suppose $X=\{x,y\}$ resolves $A$, and without loss of
  generality $x\in R_0=\{1,5,\dotsc,4k+1,4k+9,4k+13,\dotsc,8k+1,8k+5,0\}$. Then $x$ has the form
  $x=4m+1$ (where we interpret 0 as $8k+9=4(2k+2)+1$), and for all $j\in\{2,3,4\}$,
  \[d(x,j+4\ell)=
    \begin{cases}
      \ell-m+1 & \text{if }m\leq\ell\leq m+k,\\
      m-\ell & \text{if }m>\ell\geq m-k-1,\\
      m+(2k+2-\ell) &\text{if }\ell >m+k,\\
      2k-m+\ell+3 & \text{if }\ell<m-k-1.
    \end{cases}
  \]
  This implies that $\{j+4\ell\,:\,j\in\{2,3,4\}\}$ is an $S\cup\{x\}$-block, and since by
  assumption $X-\{x\}$ resolves it, Lemma~\ref{k-2resolving} implies $\lvert X-\{x\}\rvert\geq 2$,
  which is the required contradiction.
\end{proof}
\begin{lemma}
  Let $n=8k+r$, with $r\in\{5,6,7,8\}$ and
  $G=C(n,\pm\{1,2,3,4\})$. Let \[A=(\{a,a+1\},\{a+2,a+3,a+4\})\] be an $S$-cluster for some $S$. If
  $X$ resolves $A$, then $|X|\geq 3$.
\end{lemma}
\begin{proof}
  By symmetry, it is sufficient to prove the statement for $a=0$. Suppose $X=\{x,y\}$ resolves $A$,
  and without loss of generality $x \in R_0=\{1,5, \dotsc,4k+1, 0,-4,\dotsc,-4k\}$. For
  $j\in \{2,3,4\}$,
  \[
    d(x,j)=
    \begin{cases}
      1 & \text{if } x=1, \\
      m & \text{if } x=4m+1 \text{ and }1\leq m \leq k, \\
      m+1 & \text{if } x=-4m \text{ and }0\leq m \leq k.
    \end{cases}
  \]	
  Thus $\{2,3,4\}$ is an $S \cup \{x\}$-block, and Lemma~\ref{k-2resolving} implies that
  $|X-\{x\}| \geq 2$, which is the required contradiction.
\end{proof}

\begin{lemma}\label{2-4-3lemma78}
  Let $n=8k+r$, with $r\in\{7,8\}$ and $G=C(n,\pm\{1,2,3,4\})$. Let
  \[A=(\{a,a\pm1\},\{a\pm2\pm4 \ell,a\pm3\pm4\ell,a\pm4\pm4\ell\})\] be an $S$-cluster for some $S$
  with $1 \leq \ell \leq k$. If $X \subseteq V-\{a\mp4k, a\mp4(k-1),\dotsc, a\mp4(k-\ell+1)\}$
  resolves $A$, then $|X|\geq 3$.
\end{lemma}
\begin{proof}
  By symmetry we can assume $A=(\{0,1\},\{2+4 \ell,3+4\ell,4+4\ell\})$. Suppose
  $X=\{x,y\}$ resolves $A$ and $X\subset V-\{4k+r, 4k+r+4,\dotsc 4k+r+4(\ell-1)\}$.
  Then without loss of generality,
  \begin{multline*}
    x \in R_0-\{4k+r, 4k+r+4,\dotsc, 4k+r+4(\ell-1)\}\\
    =\{1,5,\dotsc,4k+1, 4k+r+4\ell,4k+r+4\ell+4,\dotsc,8k+4-4,8k+r\}.
  \end{multline*}
For $j\in\{2,3,4\}$,        
\[
  d(x,j+4\ell)=
  \begin{cases}
    \ell-m+1 & \text{if } x=1+4m \text{ and }0 \leq m \leq \ell, \\
    m-\ell & \text{if } x=1+4m\text{ and }\ell < m \leq k, \\
    m+\ell+1 & \text{if } x=-4m\text{ and }0\leq m \leq k-\ell.
  \end{cases}
\]
Thus $\{j+4 \ell : j \in \{2,3,4\}\}$ is an $S \cup \{x\}$-block, and Lemma~\ref{k-2resolving} implies
that $|X-\{x\}| \geq 2$, which is the required contradiction.
\end{proof}

\begin{lemma}\label{8lemmaAkBk}
  Let $G=C(n,\pm\{1,2,3,4\})$ be a circulant graph of order $n=8k+8$. Let
  \[A=(A_0,\,A_1,\dotsc,A_k,\,B_{1},\,B_{2},\,B_3)\]
  be an $S$-cluster where, for some $m,m'\in\{1,\dotsc,k\}$ with $m<m'$,
  \begin{align*}
    A_0 &= \{4k+4,\,4k+5,\,4k+6\},\qquad
                                    A_i = \{4k+4+4i,\,4k+5+4i\} \text{ for }1 \leq i \leq k,\\
    B_1 &= \{8k+7,\,1\},\qquad B_{2} = \{4m+2,\,4m+4\},\qquad B_{3} = \{4m'+1,\,4m'+3\}.  
  \end{align*}	
  If $X\subset V-\{0,1,\dotsc,4m'+1\}$ resolves $A$, then $|X|\geq 3$.
\end{lemma}
\begin{proof}
  Suppose $X=\{x,y\}$ resolves $A$. Without loss of generality $x$ resolves $4k+4$ and $4k+5$, i.e.,
  \[x \in R_{4k+4}=\{4k+4-4\ell,\,4k+5+4\ell\,:\,0 \leq \ell \leq k\}.\] We distinguish several cases,
  and for each possible choice of $x$ and $y$ we show that there remains an unresolved pair or
  vertices.
  \begin{description}
  \item[Case 1] $x=4k+5+4\ell;\ 1 \leq \ell \leq k$. For $S'=S\cup \{x\}$, the set
    \[B=(\{4k+5,4k+6\},\,\{8k+7,1\},\,\{4m+2,4m+4\})\]
    is an $S'$-cluster. Let $y \in R_{4k+5}=\{4k+5-4\ell_1,\,4k+6+4\ell_1\,:\, 0 \leq \ell_1 \leq k\}$.
    \begin{itemize}
    \item If $y=4k+5-4\ell_1$, $0 \leq \ell_1 \leq k-m$, then
      $d(y,4m+2)=d(y,4m+4)=k-\ell_1-m+1$. 
    \item If $y=4k+6+4\ell_1$, $0 \leq \ell_1 \leq k$, then
      $d(y,8k+7)=d(y,1)=k-\ell_1+1$. 
    \end{itemize}
  \item[Case 2] $x=4k+5$. For $S'=S\cup\{x\}$, the set
    \[B=(\{4k+8,4k+9\}, \{8k+7,1\},\{4m'+1,4m'+3\})\]
    is an $S'$-cluster. Let $y \in R_{4k+8}$.
    \begin{itemize}
    \item If $y=4k+8-4\ell_1$, $0 \leq \ell_1 \leq k-m'$, then
      $d(y,4m'+1)=d(y,4m'+3)=k-\ell_1-m'+1$. 
    \item If $y=4k+9+4\ell_1$, $0 \leq \ell_1 \leq k$, then
      $d(y,8k+7)=d(y,1)=k-\ell_1+1$. 
    \end{itemize}
  \item[Case 3] $x=4k+4$. For $S'=S\cup \{x\}$, the set
    \[B=(\{4k+5,4k+6\},\{8k+7,1\},\{4m'+1,4m'+3\})\]
    is an $S'$-cluster. Let $y \in R_{4k+5}=\{4k+5-4\ell_1,4k+6+4\ell_1: 0 \leq \ell_1 \leq k\}$.
    \begin{itemize}
    \item If $y=4k+5-4\ell_1$, $0 \leq \ell_1 \leq k-m'$, then
      $d(y,4m'+1)=d(y,4m'+3)=k-\ell_1-m'+1$. 
    \item If $y=4k+6+4\ell_1$, $0 \leq \ell_1 \leq k$, then
      $d(y,8k+7)=d(y,1)=k-\ell_1+1$. 
    \end{itemize}
  \item[Case 4] $x=4k+4-4\ell, 1 \leq \ell \leq k$. For $S'=S\cup\{x\}$ the set
    \[B=(\{4k+5,4k+6\},\{4m'+1,4m'+3\}, \{8k+8-4\ell,8k+9-4\ell\})\]
    is an $S'$-cluster. Let $y \in R_{4k+5}=\{4k+5-4\ell_1,4k+6+4\ell_1: 0 \leq \ell_1 \leq k\}$.
    \begin{itemize}
    \item If $y=4k+5-4\ell_1$, $0 \leq \ell_1 \leq k-m'$, then
      $d(y,4m'+1)=d(y,4m'+3)=k-\ell_1-m'+1$. 
    \item If $y=4k+6+4\ell_1$, $0 \leq \ell_1 \leq k$, then
      $d(y,8k+8-4\ell)=d(y,8k+9-4\ell)=k+1$.\qedhere 
    \end{itemize}
  \end{description}
\end{proof}     

\begin{lemma}\label{7lemmaAk}
  Let $G=C(n,\pm(1,2,3,4))$ be of order $n=8k+7$. Let $A=(A_0,B_0,A_1,B_1,\dotsc,A_k,B_k,A_{k+1})$,
  where $A_i=\{a+4i,a+4i+1\}$ and $B_i=\{a+2+4i,a+3+4i\}$ for $0 \leq i \leq k+1$, be an $S$-cluster
  for some $S$. If $X$ resolves $A$, then $|X|\geq 3$.
\end{lemma}
\begin{proof}
  By symmetry, it is enough to prove it for $A_i=\{1+4i,2+4i\}$ and $B_i=\{3+4i,4+4i\}$. Suppose
  $X=\{x,y\}$ resolves $A$.  Without loss of generality
  $x \in R_1=\{8k+8-4\ell,\,2+4\ell\,:\,0 \leq \ell \leq k\}$. Let $S'=S\cup\{x\}$.
  \begin{description}
  \item[Case 1] $x=8k+8-4\ell$, $0 \leq \ell \leq k$. Then
    \begin{align*}
      d(x,3+4(k-\ell)) &= d(x,4+4(k-\ell))=k+1, \\
      d(x,1+4(k+1-\ell)) &= d(x,2+4(k+1-\ell))=k+1.
    \end{align*}
    Hence, $B=(B_{k-\ell},\,A_{k+1-\ell})$ is an $S'$-cluster, an by
    Observation~\ref{observation_0123}, $\lvert X-\{x\}\rvert\geq 2$.
  \item[Case 2] $x=2+4\ell$, $0 \leq \ell \leq k$. Then
    \begin{align*}
      d(x,3+4\ell)&=d(x,4+4\ell)=1, \\
      d(x,5+4\ell)&=d(x,6+4\ell)=1.
    \end{align*}
    Hence, $B=(B_{\ell}, A_{\ell+1})$ is an $S'$-cluster, and by Observation~\ref{observation_0123},
    $\lvert X-\{x\}\rvert\geq 2$.\qedhere
  \end{description}
\end{proof}

\begin{lemma}\label{7lemmaAkBk}
  Let $G=C(n;\pm\{1,2,3,4\})$ be of order $n=8k+7$. Let $A=(A_0,A_1,\dotsc,A_k,B_{1},B_{2})$ be an
  $S$-cluster, where, for some $m' \in\{1,\dotsc,k\}$,
  \begin{align*}
	A_i &= \{3+4i,\,4+4i\}\text{ for } 0 \leq i \leq k-1, & A_k &= \{3+4k,\,4+4k,\,5+4k\},\\
	B_{1}&=\{3+4(k+m'),\,4+4(k+m')\},& B_{2} &= \{5+4(k+m'),\,6+4(k+m')\}.
  \end{align*}
  If $X$ resolves $A$, then $\lvert X\rvert\geq 3$.
\end{lemma}
\begin{proof}
  Suppose $X=\{x,y\}$ resolves $A$. Without loss of generality
  \[x \in R_3=\{4+4\ell,\, 8k+10-4\ell\,:\, 0 \leq \ell \leq k\}.\]
  We distinguish the following cases.
  \begin{description}
  \item[Case 1] $x=4+4\ell$, $0\leq\ell\leq k$.  For all $u \in B_1$ and all $v \in
    B_2$,
    \begin{align*}
      d(x,u) &=
               \begin{cases}
                 k+m'-\ell & \text{if } \ell \geq m', \\
                 k+\ell-m'+2&\text{if } \ell < m',
               \end{cases} &
                             d(x,v) &=
                                      \begin{cases}
                                        k+m'-\ell+1 & \text{if } \ell \geq m', \\
                                        k+\ell-m'+2&\text{if } \ell < m'.
                                      \end{cases}
    \end{align*}
    Hence, $B=(B_1,\,B_2)$ is an $S'$-cluster for $S'=S\cup\{x\}$.  By
    Observation~\ref{observation_0123}, $\lvert X-\{x\}\rvert\geq 2$.
  \item[Case 2] $x=8k+10-4\ell$, $0 \leq \ell \leq k$. For $S'=S\cup\{x\}$,
    \[B = (B_1,\,A_k-\{3+4k\}) = (\{3+4k+4m',4+4k+4m'\},\,\{4+4k,5+4k\})\] is an
    $S'$-cluster. Without loss of generality
    $y \in R_{4+4k}=\{4+4k-4\ell_1,\,5+4k+4\ell_1\,:\,0 \leq \ell_1 \leq k\}$. If $y=4+4k-4\ell_1$,
    $0 \leq \ell_1 \leq k$, then
    \[d(y,3+4k+4m')=d(y,4+4k+4m')=
      \begin{cases}
        m'+\ell_1& \text{if } m'+\ell_1 \leq k+1,\\
        2k+2-(m'+\ell_1) & \text{if } m'+\ell_1 >k+1.
      \end{cases}\]
    If $y=5+4k+4\ell_1$, $0 \leq \ell_1 \leq k$, then
    \[d(y,3+4k+4m')=d(y,4+4k+4m')=
      \begin{cases}
        m'-\ell_1& \text{if } \ell_1 < m',\\
        \ell_1-m'+1 & \text{if } \ell_1 \geq m'.
      \end{cases}\]
    In both cases, $B_1$ is an $S\cup\{x,y\}$-block, which is the required
    contradiction.\qedhere
  \end{description}
\end{proof}
\begin{lemma}\label{2-22-3lemma7}
  Let $G=C(n, \pm \{1,2,3,4\})$ with $n=8k+7$, and let
  \[A=(\{1,2\},\{4k+2,\,4k+3\},\{4k+4,\,4k+5\},\{4(k+\ell)+j\,:\,j \in\{7,8,9\}\})\]
  be an $S$-cluster with $\ell\in\{0,1,\dotsc,k\}$. If $X\subset V-\{2,6,\dotsc,2+4\ell\}$ resolves $A$, then
  $\lvert X\rvert\geq 3$.
\end{lemma}
\begin{proof}
  Suppose $X=\{x,y\}\subset V-\{2,6,\dotsc,2+4\ell\}$ is a resolving set of $A$. Without loss of
  generality, 
  \[ x \in R_1-\{2,6,\dotsc,2+4\ell\}=\{8k+8-4m:0\leq m \leq
      k\}\cup\{2+4m:\ell+1\leq m \leq k\}.\]
\begin{description}
\item[Case 1] $x=1$. From $d(1,4k+2)=d(1,4k+3)=k+1$ and $d(1,4k+4)=d(1,4k+5)=k+1$ it follows that
  $A'=(\{4k+2,\,4k+3\},\,\{4k+4,\,4k+5\})$ is an $S\cup\{x\}$-cluster, and by
  Observation~\ref{observation_0123}, $|X-\{x\}|\geq 2$.
\item[Case 2] $x=8k+8-4m$, $1\leq m \leq k$. From $d(x,4k+2)=d(x,4k+3)=k-m+2$ and
  $d(x,4k+4)=d(x,4k+5)=k-m+1$ it follows that $A'=(\{4k+2,\,4k+3\},\,\{4k+4,\,4k+5\})$ is an
  $S\cup\{x\}$-cluster, and by Observation~\ref{observation_0123}, $|X-\{x\}|\geq 2$.
\item[Case 3] $x=2+4m$, $\ell+1 \leq m \leq k$. From
  $d(x,4k+4\ell+9)=d(x,4k+4\ell+8)=d(x,4k+4\ell+7)=k+l-m+2$ it follows that
  $A'=\{4k+4\ell+j\,:\,j\in \{7,8,9\}\}$ is an $S\cup\{x\}$-cluster, and by
  Lemma~\ref{k-2resolving}, $|X-\{x\}|\geq 2$.\qedhere
\end{description}
\end{proof}
\begin{lemma}\label{2-4-3lemma_r56}
  Let $n=8k+r$, with $r \in\{2,5,6\}$, $G=C(n,\pm\{1,2,3,4\})$, and let
  \[A=\left(\{a,a\pm1\},\,\{a\pm j\pm4\ell\,:\,j\in\{2,3,4\}\}\right)\]
  be an $S$-cluster for some
  $S$, where $0 \leq \ell \leq k$. If $X \subseteq V$ resolves $A$, then $\lvert X\rvert\geq 3$.
\end{lemma}
\begin{proof}
  By symmetry, we can assume $A=\left(\{0,1\},\,\{j+4\ell\,:\,j\in\{2,3,4\}\}\right)$. Suppose
  $X=\{x,y\}$ resolves $A$, and without loss of generality $x\in R_0=\{1+4m,\,-4m\,:\,0 \leq m \leq
  k\}$. For all $j\in\{2,3,4\}$, $x=1+4m$ and $0\leq m \leq k$,
  \[d(x,j+4\ell)=
    \begin{cases}
      \ell-m+1 & \text{if }x=1+4m \text{ and } m\leq\ell,\\
      m-\ell & \text{if }x=1+4m\text{ and }\ell+1\leq m\leq k,\\
      \ell+m+1 & \text{if }x=-4m \text{ and }\ell+m\leq\ k,\\
      2k-m-\ell & \text{if }x=-4m\text{ and }\ell+ m>k\text{ and } r=2,\\
      2k-m-\ell+1 & \text{if }x=-4m\text{ and }\ell+ m>k\text{ and } r\in \{5,6\}.
    \end{cases}
  \]
  Consequently, $\{j+4\ell\,:\,j\in\{2,3,4\}\}$ is an $S\cup\{x\}$-block, and
  Lemma~\ref{k-2resolving} implies $\lvert X-\{x\}\rvert\geq 2$, which is the required
  contradiction.
\end{proof}
\begin{observation}\label{r5_observation_0156}
  Let $G=C(n,\pm\{1,2,3,4\})$ with $n=8k+r$ and $r\in\{2,5\}$. Let $A=\{a,a+1,a+5,a+6\}$. Then
  $\lvert X\rvert\geq 2$ for every set $X$ that resolves $A$.
\end{observation}
\begin{proof}
  By symmetry we can assume $A=\{0,1,5,6\}$. Suppose $X=\{x\}$ resolves
  $A$. Then $x\in R_0\cap R_5$ and with
  \begin{align*}
  R_0 &=\{1,5, \dotsc,4k+1\}\cup\{0,8k+r-4,\dotsc,4k+r\},\\
  R_5 &=\{6,10, \dotsc,4k+6\}\cup\{5,1,8k+r-3,\dotsc,4k+r+5\},  
  \end{align*}
it follows that $x \in \{1,5\}$. If $x=1$, then $d(x,0)=d(x,5)=1$; and if $x=5$, then
$d(x,1)=d(x,6)=1$. In both cases we obtain the required contradiction.
\end{proof}
\begin{observation}\label{r5_observation_025_67}
  Let $G=C(n,\pm\{1,2,3,4\})$ with $n=8k+r$ and $r\in
  \{2,5\}$. Let \[A=(\{a,a+2,a+5\},\{a+6,a+7\})\] be an $S$-cluster for some $S$. If $X\subset V$
  resolves $A$, then $\lvert X\rvert\geq 3$.
\end{observation}
\begin{proof}
  It is enough to prove it for $A=(\{0,2,5\},\{6,7\})$. Suppose $X=\{x\}$ resolves $A$. Then
  $x \in R_{6}=\{7+4m,\,6-4m\,:\, 0\leq m \leq k\}$.
\begin{description}
\item[$n=8k+5$] In this case we have the following possibilities for $x$.
\begin{itemize}
	\item  If $x=7+4m,$ $0\leq m \leq k-1$, then $d(x,0)=d(x,2)=m+2$.
	\item  If $x=7+4k$, then $d(x,0)=d(x,2)=k$.
	\item If $x=6$, then $d(x,2)=d(x,5)=1$. 
	\item If $x=2$, then $d(x,0)=d(x,5)=1$.
	\item If $x=6-4m,$ $2 \leq m \leq k$, then $d(x,0)=d(x,2)=m-1$. 
\end{itemize}
\item[$n=8k+2$] In this case we have the following possibilities for $x$.
\begin{itemize}
	\item If $x=7+4m,$ $0\leq m \leq k-2$, then $d(x,0)=d(x,2)=m+2$.
	\item If $x=3+4k$, then $d(x,0)=d(x,5)=k$.
	\item If $x=7+4k$, then $d(x,2)=d(x,5)=k$.
	\item If $x=6$, then $d(x,2)=d(x,5)=1$. 
	\item If $x=2$, then $d(x,0)=d(x,5)=1$. 
	\item If $x=6-4m,$ $2 \leq m \leq k$, then $d(x,0)=d(x,2)=m-1$.\qedhere
\end{itemize}
\end{description}
\end{proof}

\begin{lemma}\label{r5_lemma_01}
  Let $G=C(n,\pm\{1,2,3,4\})$ with $n= 8k+r$, $r\in\{2,5\}$. If $X\subseteq V$ resolves
  \[A=\{a,\,a+1,\,a+2,\,a+5,\,a+6,\,a+7\}.\]
  then $\lvert X\rvert\geq 3$.
\end{lemma}
\begin{proof}	
  By symmetry, we can assume $A=\{0,1,2,5,6,7\}$. Suppose $X=\{x,y\}$ resolves $A$, and let
  $x \in R_0=\{-4m,\,1+4m\,:\, 0 \leq m \leq k\}$. Let $S'=S\cup\{x\}$. In
  Table~\ref{tab:r5_lemma_01} we summarize the possible choices for $x$, and in each case we exhibit
  an $S'$-block or $S'$-cluster which implies that $\lvert X-\{x\}\rvert\geq 2$, which is the
  required contradiction:
\begin{itemize}
\item For the $S'$-block $\{5,6,7\}$, $\lvert X-\{x\}\rvert\geq 2$ follows from
  Lemma~\ref{k-2resolving}.
\item For the $S'$-block $\{1,2,6,7\}$, $\lvert X-\{x\}\rvert\geq 2$ follows from
  Observation~\ref{r5_observation_0156}.
\item For the $S'$-block $(\{0,2,5\},\{6,7\})$, $\lvert X-\{x\}\rvert\geq 2$ follows from
  Observation~\ref{r5_observation_025_67}.\qedhere
\end{itemize}
\begin{table}[htb]
  \centering
  \caption{Cases in the proof of Lemma~\ref{r5_lemma_01}. The last column contains the $S'$-block or
    $S'$-cluster which implies $\lvert X-\{x\}\rvert\geq 2$.}
  \label{tab:r5_lemma_01}
  \begin{mytabularwrap}
    \begin{tabular}{cccccccccc}
      \toprule
      &$x$& $d(x,0)$   & $d(x,1)$ & $d(x,2)$ & $d(x,5)$ & $d(x,6)$ & $d(x,7)$ &    Witness\\ 
      \midrule
      \multirow{5}{*}{\begin{turn}{90}$r=5$\end{turn}}	& $-4m,$ $0 \leq m \leq k-1$&$m$   & $m+1$ & $m+1$ & $m+2$ & $m+2$ & $m+2$ &   $\{5,6,7\}$\\ 
      & $4k+5$	&$k$   & $k+1$ & $k+1$ & $k$ & $k$   & $k$   &  $\{5,6,7\}$\\ 
      & $1+4m$, $2 \leq m \leq k$&$m+1$   & $m$   &  $m$  &  $m-1$  & $m-1$   & $m-1$    &   $\{5,6,7\}$\\ 
      & $5$			&$2$   & $1$   &  $1$  &  $0$  & $1$   & $1$      & $\{1,2,6,7\}$\\ 
      & $1$			&$1$   & $0$   &   1   &   1   & 2   & $2$   &
                                                                               $(\{0,2,5\},\{6,7\})$\\
      \midrule
      \multirow{6}{*}{\begin{turn}{90}$r=2$\end{turn}}	& $8k+2-4m,$ $0\leq m \leq k-2$  &  $m$  & $m+1$ & $m+1$ & $m+2$ & $m+2$ & $m+2$ &    $\{5,6,7\}$   \\
      &        $4k+6$          & $k-1$ &  $k$  &  $k$  &  $k+1$  & $k$ &  $k$  &  $\{1,2,6,7\}$ \\
      &         $4k+2$           &  $k$  & $k+1$ &  $k$  &  $k$  &  $k-1$  & $k-1$ &   $(\{0,2,5\},\{6,7\})$   \\
      &$1+4m,$ $2 \leq m \leq k$ & $m+1$ &  $m$  &  $m$  &  $m-1$  & $m-1$ & $m-1$ &   $\{5,6,7\}$   \\
      &          $5$           &  $2$  &  $1$  &  $1$  &  $0$  &  $1$  &  $1$  &   $\{1,2,6,7\}$ \\
      &          $1$           &  $0$  &  $1$  &   1   &   1   &   2   &  $2$  &    $(\{0,2,5\},\{6,7\})$   \\ \bottomrule    
    \end{tabular}
  \end{mytabularwrap}
\end{table}	
\end{proof}
\begin{lemma}\label{r5_lemma_02}
  Let $G=C(n,\pm\{1,2,3,4\})$ with $n= 8k+r$, $r\in\{2,5\}$, and let
  \[A=(\{a,\, a+1\},\, \{a+2,\, a+3,\, a+5,\, a+7,\, a+8\})\] be an $S$-cluster for some $S$. If
  $X\subseteq V$ resolves $A$ then $|X|\geq 3$.
\end{lemma}
\begin{proof}	
  By symmetry we can assume $A=(\{0,1\},\,\{2,3,5,7,8\})$. Suppose $X=\{ x, y \}$ resolves $A$, and
  without loss of generality $x \in R_0=\{-4m,\,1+4m\,:\, 0 \leq m \leq k\}$. Let $S'=S\cup\{x\}$. In
  Table~\ref{tab:r5_lemma_02} we summarize the possible choices for $x$, and in each case we exhibit
  an $S'$-block which implies that $\lvert X-\{x\}\rvert\geq 2$, which is the required
  contradiction:
  \begin{itemize}
  \item For the $S'$-blocks $\{5,7,8\}$ and $\{2,3,5\}$, $\lvert X-\{x\}\rvert\geq 2$ follows from
    Lemma~\ref{k-2resolving}.
  \item For the $S'$-block $\{2,3,7,8\}$, $\lvert X-\{x\}\rvert\geq 2$ follows from
    Observation~\ref{r5_observation_0156}.\qedhere
  \end{itemize}
  \begin{table}[htb]
    \centering
    \caption{Cases in the proof of Lemma~\ref{r5_lemma_02}. The last column contains the $S'$-block
      which implies $\lvert X-\{x\}\rvert\geq 2$.}
    \label{tab:r5_lemma_02}
    \begin{mytabularwrap}
      \begin{tabular}{cccccccccc}
        \toprule
        &$x$& $d(x,0)$   & $d(x,1)$ & $d(x,2)$ & $d(x,3)$ & $d(x,5)$ & $d(x,7)$ & $d(x,8)$ &      $S'$-block\\ 
        \midrule
        \multirow{5}{*}{\begin{turn}{90}$r=5$\end{turn}}	& $-4m$, $0 \leq m \leq k-1$&$m$   & $m+1$ & $m+1$ & $m+1$ & $m+2$ & $m+2$ & $m+2$ &   $\{5,7,8\}$\\ 
        & $4k+5$	&$k$   & $k+1$ & $k+1$ & $k+1$ & $k$   & $k$   & $k$   &   $\{5,7,8\}$\\ 
        & $1+4m$, $2 \leq m \leq k$&$m+1$   & $m$   &  $m$  &  $m$  & $m-1$   & $m-1$   & $m-1$   &   $\{5,7,8\}$\\ 
        & $5$			&$2$   & $1$   &  $1$  &  $1$  & $0$   & $1$   & $1$   & $\{2,3,7,8\}$\\ 
        & $1$			&$1$   & $0$   &   1   &   1   & 1   & $2$   & $2$   &
                                                                                       $\{2,3,5\}$\\
        \midrule
        \multirow{6}{*}{\begin{turn}{90}$r=2$\end{turn}}	& $-4m$, $0\leq m \leq k-2$  &  $m$  & $m+1$ & $m+1$ & $m+1$ & $m+2$ & $m+2$ & $m+2$ &   $\{5,7,8\}$   \\
        &        $4k+6$          & $k-1$ &  $k$  &  $k$  &  $k$  & $k+1$ &  $k$  &  $k$  & $\{2,3,7,8\}$ \\
        &         $4k+2$           &  $k$  & $k+1$ &  $k$  &  $k$  &  $k$  & $k-1$ & $k-1$ &   $\{2,3,5\}$   \\
        &$1+4m$, $2 \leq m \leq k$ & $m+1$ &  $m$  &  $m$  &  $m$  & $m-1$ & $m-1$ & $m-1$ &   $\{5,7,8\}$   \\
        &          $5$           &  $2$  &  $1$  &  $1$  &  $1$  &  $0$  &  $1$  &  $1$  & $\{2,3,7,8\}$ \\
        &          $1$           &  $1$  &  $0$  &   1   &   1   &   1   &  $2$  &  $2$  &   $\{2,3,5\}$   \\ \bottomrule    
      \end{tabular}
    \end{mytabularwrap}
\end{table}
\end{proof}
\begin{observation}\label{r5_observation_01257}
  Let $n=8k+r$, with $r\in\{2,5\}$ and let $A=(\{a,\,a+1\},\,\{a+2,\,a+5,\,a+7\})$ be an $S$-cluster for some
  $S$. If $X\subseteq V$ resolves $A$ then $|X|\geq 2$.
\end{observation}
\begin{proof}
  By symmetry we can assume $A=(\{0,1\},\,\{2,5,7\})$. Suppose $X=\{x\}$ is a resolving set of $A$. Then
  $x \in R_0=\{8k+r-4m,1+4m: 0\leq m \leq k\}$.

  For $n=8k+5$ we have the following possibilities:
    \begin{itemize}
    \item If $x=8k+5-4m$, $0 \leq m \leq k-1$, then $d(x,5)=d(x,7)=m+2$.
    \item If $x=4k+5$, then $d(x,5)=d(x,7)=k$.
    \item If $x=1$, then $d(x,2)=d(x,5)=1$.
    \item If $x=5$, then $d(x,2)=d(x,7)=1$.
    \item If $x=1+4m$, $2 \leq m \leq k$, then $d(x,5)=d(x,7)=m-1$.
    \end{itemize}

  For $n=8k+2$ we have the following possibilities:
    \begin{itemize}
    \item If $x=8k+2-4m$, $0 \leq m \leq k-2$, then $d(x,5)=d(x,7)=m+2$.
    \item If $x=4k+6$, then $d(x,2)=d(x,7)=k$.
    \item If $x=4k+2$, then $d(x,2)=d(x,5)=k$.
    \item If $x=1$, then $d(x,2)=d(x,5)=1$.
    \item If $x=5$, then $d(x,2)=d(x,7)=1$.
    \item If $x=1+4m$, $2 \leq m \leq k$, then $d(x,5)=d(x,7)=m-1$.\qedhere
    \end{itemize}  
\end{proof}
\begin{observation}\label{r5_observation_023568}
  Let $n=8k+r$ with $r \in \{2,5\}$ and let $A=(\{a,a+2\},\,\{a+3,a+5\},\,\{a+6,a+8\})$ be an
  $S$-cluster for some $S$. If $X\subseteq V$ resolves $A$ then $\lvert X\rvert\geq 2$.
\end{observation}
\begin{proof}
  By symmetry we can assume $A=(\{0,2\},\,\{3,5\},\,\{6,8\})$. Suppose $x$ resolves $A$. Then
  $x \in R_0\cup R_1$, where $R_0=\{8k+r-4m,\,1+4m\,:\,0\leq m \leq k\}$ and
  $R_1=\{8k+r+1-4m,\,2+4m\,:\,0\leq m \leq k\}$.
  For $n=8k+5$ we have the following possibilities:
  \begin{itemize}
  \item If $x=8k+5-4m, \ 0\leq m \leq k$, then $d(x,6)=d(x,8)=
    \begin{cases}
      m+2& \text{when } 0 \leq m \leq k-1,\\
      k& \text{when } m=k.
    \end{cases}$
  \item If $x=1+4m,\ 0 \leq m \leq k$, then $d(x,6)=d(x,8)=
    \begin{cases}
      m-1& \text{when }2 \leq m \leq k,\\
      2-m& \text{when } m \in \{0,1\}.
    \end{cases}$
  \item If $x=8k+6-4m$, $0 \leq m \leq k$, then $d(x,3)=d(x,5)=m+1$.
  \item If $x=2+4m,\ 0 \leq m \leq k$, then $d(x,3)=d(x,5)=
    \begin{cases}
      m &\text{when }1 \leq m \leq k,\\
      1&\text{when }m=0.
    \end{cases}$
 \end{itemize}
  For $n=8k+2$ we have the following possibilities:
  \begin{itemize}
  \item If $x=8k+2-4m,\ 0\leq m \leq k$, then $d(x,6)=d(x,8)=\begin{cases}
      m+2& \text{when } 0 \leq m \leq k-2,\\
      2k-m-1& \text{when } m\in \{k-1,k\}.
    \end{cases}$
  \item If $x=1+4m,\ 0 \leq m \leq k$, then $d(x,6)=d(x,8)=\begin{cases}
      m-1& \text{when }2 \leq m \leq k,\\
      2-m& \text{when } m \in \{0,1\}.
    \end{cases}$\\
  \item If $x=8k+6-4m,\ 0 \leq m \leq k$, then $d(x,3)=d(x,5)=\begin{cases}
      m+1& \text{when } 0 \leq m \leq k-1,\\
      k& \text{when }m=k.
    \end{cases}$.
  \item If $x=2+4m,\ 0 \leq m \leq k$, then $d(x,3)=d(x,5)=\begin{cases}
      m &\text{when }1 \leq m \leq k,\\
      1&\text{when }m=0.
    \end{cases}$\qedhere
  \end{itemize}
\end{proof}
\begin{lemma}\label{r5lemma_03}
  Let $G=C(n,\pm\{1,2,3,4\})$ with $n= 8k+r$ with $r\in\{2,5\}$, and let
  \[A=(\{a, a+1, a+2\},\,\{a+3, a+5, a+6, a+8\})\] be an $S$-cluster for some $S$. If $X\subseteq V$
  of $A$, $\lvert X\rvert\geq 3$.
\end{lemma}
\begin{proof}
  By symmetry we can assume $A=(\{0,1,2\},\,\{3,5,6,8\})$. Suppose $X=\{ x, y \}$ resolves
  $A$. Without loss of generality $x \in R_0=\{-4m,\,1+4m\,:\,0 \leq m \leq k\}$.  Let
  $S'=S\cup\{x\}$. In Table~\ref{tab:r5_lemma_03} we summarize the possible choices for $x$, and in
  each case we exhibit an $S'$-block or $S'$-cluster implying that $\lvert X-\{x\}\rvert\geq 2$,
  which is the required contradiction:
  \begin{itemize}
  \item For the $S'$-block $\{5,7,8\}$, $\lvert X-\{x\}\rvert\geq 2$ follows from
    Lemma~\ref{k-2resolving}.
  \item For the $S'$-cluster $(\{1,2\},\,\{3,6,8\})$, $\lvert X-\{x\}\rvert\geq 2$ follows from
    Observation~\ref{r5_observation_01257}.
  \item For the $S'$-cluster $(\{0,2\},\,\{3,5\},\,\{6,8\})$, $\lvert X-\{x\}\rvert\geq 2$ follows from
    Observation~\ref{r5_observation_023568}.\qedhere
  \end{itemize}
  \begin{table}[htb]
    \centering
    \caption{Cases in the proof of Lemma~\ref{r5lemma_03}. The last column contains the $S'$-block
      or $S'$-cluster which implies $\lvert X-\{x\}\rvert\geq 2$.}
    \label{tab:r5_lemma_03}
    \begin{mytabularwrap}
      \begin{tabular}{cccccccccc}
        \toprule
        &$x$& $d(x,0)$   & $d(x,1)$ & $d(x,2)$ & $d(x,3)$ & $d(x,5)$ & $d(x,6)$ & $d(x,8)$ &      Witness\\ 
        \midrule
        \multirow{5}{*}{\begin{turn}{90}$r=5$\end{turn}}	& $-4m,\ 0 \leq m \leq k-1$&$m$   & $m+1$ & $m+1$ & $m+1$ & $m+2$ & $m+2$ & $m+2$ &   $\{5,6,8\}$\\ 
        & $4k+5$	&$k$   & $k+1$ & $k+1$ & $k+1$ & $k$   & $k$   & $k$   &   $\{5,6,8\}$\\ 
        & $1+4m,\ 2 \leq m \leq k$&$m+1$   & $m$   &  $m$  &  $m$  & $m-1$   & $m-1$   & $m-1$   &   $\{5,6,8\}$\\ 
        & $5$			&$2$   & $1$   &  $1$  &  $1$  & $0$   & $1$   & $1$   & $(\{1,2\},\{3,6,8\})$\\ 
        & $1$			&$1$   & $0$   &   1   &   1   & 1   & $2$   & $2$   &
                                                                                       $(\{0,2\},\{3,5\},\{6,8\})$\\
        \midrule
        \multirow{6}{*}{\begin{turn}{90}$r=2$\end{turn}}	& $-4m,\ 0\leq m \leq k-2$  &  $m$  & $m+1$ & $m+1$ & $m+1$ & $m+2$ & $m+2$ & $m+2$ &   $\{5,6,8\}$   \\
        &        $4k+6$          & $k-1$ &  $k$  &  $k$  &  $k$  & $k+1$ &  $k$  &  $k$  & $(\{1,2\},\{3,6,8\})$ \\
        &         $4k+2$           &  $k$  & $k+1$ &  $k$  &  $k$  &  $k$  & $k-1$ & $k-1$ &   $(\{0,2\},\{3,5\},\{6,8\})$   \\
        &$1+4m,\ 2 \leq m \leq k$ & $m+1$ &  $m$  &  $m$  &  $m$  & $m-1$ & $m-1$ & $m-1$ &   $\{5,6,8\}$   \\
        &          $5$           &  $2$  &  $1$  &  $1$  &  $1$  &  $0$  &  $1$  &  $1$  & $(\{1,2\},\{3,6,8\})$ \\
        &          $1$           &  $1$  &  $0$  &   1   &   1   &   1   &  $2$  &  $2$  &   $(\{0,2\},\{3,5\},\{6,8\})$   \\ \bottomrule    
		\end{tabular}
	\end{mytabularwrap}
\end{table}
\end{proof}
In the next six lemmas we show that for $n\equiv 3\pmod 8$ certain $S$-blocks or $S$-clusters cannot
be resolved by a single vertex $x$. In the proofs we always assume $a=0$, and we verify that
for every possible $x$ there remains a pair of unresolved vertices.
\begin{lemma}\label{mod3lemma2-3-2}
  For $n=8k+3$, if $X\subseteq V$ resolves $A=\{a,\,a+1,\,a+5,\,a+6\}$
  then $\lvert X\rvert\geq 2$.
\end{lemma}
\begin{proof}	
  If $x\in V$ resolves $A$ then $x\in R_0=\{8k+3-4m,\,1+4m\,:\,0 \leq m \leq k\}$.
  \begin{itemize}
  \item If $x=8k+3-4m,\ 0\leq m \leq k-1$, then $d(x,5)=d(x,6)=m+2$.
  \item If $x=4k+3$, then $d(x,5)=d(x,6)=k$.
  \item If $x=1+4m,\ 2 \leq m \leq k$, then $d(x,5)=d(x,6)=m-1$.
  \item If $x=5$, then $d(x,1)=d(x,6)=1$.
  \item If $x=1$, then $d(x,0)=d(x,5)=1$. \qedhere
  \end{itemize}
\end{proof}
\begin{lemma}\label{mod3lemma2-7-2a}
  For $n=8k+3$, if $X$ resolves $A=(\{a,a+1\},\,\{a+2,a+5,a+7\},\,\{a+9,a+10\})$ then
  $\lvert X\rvert\geq 2$.
\end{lemma}
\begin{proof}
  If $x\in V$ resolves $A$ then $x\in R_0=\{8k+3-4m,\,\ 1+4m\,:\,0\leq m\leq k\}$.
  \begin{itemize}
  \item If $x=8k+3-4m,\ 0\leq m \leq k-2$, then $d(x,5)=d(x,7)=m+2$.
  \item If $x=4k+7$, then $d(x,2)=d(x,7)=k$.
  \item If $x=4k+3$, then $d(x,9)=d(x,10)=k-1$.
  \item If $x=1+4m,\ 3 \leq m \leq k$, then $d(x,5)=d(x,7)=m-1$.
  \item If $x=9$, then $d(x,5)=d(x,7)=1$.
  \item If $x=5$, then $d(x,2)=d(x,7)=1$.
  \item If $x=1$, then $d(x,2)=d(x,5)=1$.\qedhere
  \end{itemize}
\end{proof}
\begin{lemma}\label{mod3lemma2-5-2}
  For $n=8k+3$, in $X$ resolves $A=(\{a,a+1\},\,\{a+7,a+8\})$ then $\lvert X\rvert\geq 2$.
\end{lemma}
\begin{proof}
  If $x\in V$ resolves $A$ then $x\in R_0=\{8k+3-4m,\,1+4m\,:\,0 \leq m \leq k\}$.
  \begin{itemize}
  \item If $x=8k+3-4m$, $0\leq m \leq k-2$, then $d(x,7)=d(x,8)=m+2$.
  \item If $x=4k+7$, then $d(x,7)=d(x,8)=k$.
  \item If $x=4k+3$, then $d(x,7)=d(x,8)=k-1$.
  \item If $x=1+4m,\ 2 \leq m \leq k$, then $d(x,7)=d(x,8)=m-1$.
  \item If $x=5$, then $d(x,7)=d(x,8)=1$.
  \item If $x=1$, then $d(x,7)=d(x,8)=2$.\qedhere
  \end{itemize}
\end{proof}
\begin{lemma}\label{mod3lemma222}
  For $n=8k+3$, if $X$ resolves $A=(\{a,a+1\},\, \{a+2,a+3\},\, \{a+4,a+5\})$ then $\lvert X\rvert\geq 2$.
\end{lemma}
\begin{proof}
  If $x\in V$ resolves $A$ then $x\in R_0=\{8k+3-4m,1+4m: 0 \leq m \leq k\}$.
  \begin{itemize}
  \item If $x=8k+3-4m,\ 0\leq m \leq k-1$, then $d(x,2)=d(x,3)=m+1$.
  \item If $x=4k+3$, then $d(x,4)=d(x,5)=k$.
  \item If $x=1+4m,\ 1 \leq m \leq k$, then $d(x,2)=d(x,3)=m$.
  \item If $x=1$, then $d(x,2)=d(x,3)=1$.\qedhere
  \end{itemize}
\end{proof}
\begin{lemma}\label{mod3lemma2-7-2b}
  For $n=8k+3$, if $X$ resolves $A=(\{a,a+1\},\,\{a+3,a+5,a+8\})$ then $|X|\geq 2$.
\end{lemma}
\begin{proof}
  If $x\in V$ resolves $A$ then $x\in R_0=\{8k+3-4m,\,1+4m\,:\,0 \leq m \leq k\}$.
  \begin{itemize}
  \item If $x=8k+3-4m,\ 0\leq m \leq k-2$, then $d(x,5)=d(x,8)=m+2$.
  \item If $x=4k+7$, then $d(x,3)=d(x,8)=k$.
  \item If $x=4k+3$, then $d(x,3)=d(x,5)=k$.
  \item If $x=1+4m,\ 2 \leq m \leq k$, then $d(x,5)=d(x,8)=m-1$.
  \item If $x=5$, then $d(x,3)=d(x,8)=1$.
  \item If $x=1$, then $d(x,3)=d(x,5)=1$.\qedhere
  \end{itemize}
\end{proof}

\begin{lemma}\label{mod3lemma2-1-2}
  For $n=8k+3$, if $X$ resolves $A=(\{a,a+1\},\,\{a+3,a+4\})$ then $|X|\geq 2$.
\end{lemma}
\begin{proof}
  If $x\in V$ resolves $A$ then $x\in R_0=\{8k+3-4m,\ 1+4m: 0 \leq m \leq k\}$.
  \begin{itemize}
  \item If $x=8k+3-4m,\ 0\leq m \leq k-1$, then $d(x,3)=d(x,4)=m+1$.
  \item If $x=4k+3$, then $d(x,3)=d(x,4)=k$.
  \item If $x=1+4m,\ 1 \leq m \leq k$, then $d(x,3)=d(x,4)=m$.
  \item If $x=1$, then $d(x,3)=d(x,4)=1$.\qedhere
  \end{itemize}
\end{proof}
\section{Lower bounds}\label{sec:proof}
In this section we prove the lower bounds for Theorem~\ref{mainthm}, using the lemmas proved in Section~\ref{sec:auxiliary}. 
\subsection{$n \equiv 1 \pmod 8$}
We write $n=8k+9$, $G=C(n,\pm\{1,2,3,4\})$, and prove $\dim(G)\geq 6$. Suppose
$B=\{w_1,w_2,w_3,w_4,w_5\}$ is a metric basis. Without loss of generality $w_1=0$ and
$w_2,w_3\in\{1,2,\dotsc,4k+3\}$. We make a case distinction with respect to the possibilities for
the set $S=\{0,w_2,w_3\}$. By Lemma~\ref{lem:min_dist_789} we may assume that
$w_2\geq 4$, $w_3\geq 4$, and $\lvert w_2-w_3\rvert\geq 4$.
\begin{description}
\item[Case 1] $w_2\equiv w_3\pmod 4$, that is $S=\{0,\,4m+i,\,4m'+i\}$ with $i\in\{0,1,2,3\}$ and
  $1\leq m<m'\leq k$. The vertex set $A=\{4k+i+1,\,4k+i+2,\,4k+i+3,\,4k+i+4\}$ is an $S$-block with
  representing vector $r(A\vert S)=(k+1,k-m+1,k-m'+1)$. Since $B-S$ resolves $A$,
  Lemma~\ref{k-2resolving} implies $\lvert B-S\rvert\geq 3$.
\item[Case 2] $S=\{0,\,4m,\,4m'+i\}$ with $m\in\{1,\dotsc,k+1\}$, $m'\in\{1,\dotsc,k\}$, $m\neq m'$ and
  $i\in\{1,2,3\}$. In this case $A=\{4(m'+k)+\ell\ :\ \ell\in\{5,6,7,8\}\}$ is an $S$-block with
  \[r(A\vert S)=
    \begin{cases}
      (k-m'+1,\ k-m'+m+1,\ k+1) & \text{if }m<m',\\
      (k-m'+1,\ k-(m-m')+2,\ k+1) & \text{if }m>m'.
    \end{cases}
  \]
  By Lemma~\ref{k-2resolving}, $\lvert B-S\rvert\geq 3$.
\item[Case 3] $S=\{0,\,4m+1,\,4m'+i\}$ with $m,m'\in\{1,\dotsc,k\}$, $m\neq m'$ and $i\in\{2,3\}$.  Let
  $A=(A_1,A_2)$ with $A_1=\{4k+4,4k+5\}$ and $A_2=\{4k+4+j+4m'\,:\,j\in\{2,3,4\}\}$. Then $A$
  is an $S$-cluster with $r(A_1\vert S)=(k+1,k-m+1,k-m'+1)$ and
  \[r(A_2\vert S)=
    \begin{cases}
      (k-m'+1,\ k-m'+m+1,\ k+1) & \text{if }m<m',\\
      (k-m'+1,\ k-(m-m')+2,\ k+1) & \text{if }m>m'.
    \end{cases}
  \]
  By Lemma~\ref{2-4-3lemma_1}, $\lvert B-S\rvert\geq 3$.
\item[Case 4] $S=\{0,4m+2,4m'+3\}$ with $m,m'\in\{1,\dotsc,k\}$, $m\neq m'$. Let $A=(A_1,A_2)$
  with $A_1=\{3,4\}$ and $A_2=\{4(m+k)+\ell\,:\,\ell\in\{5,6,7\}\}$. Then $A$ is an $S$-cluster with
  $r(A_1\vert S)=(1,m,m')$ and
  \[r(A_2\vert S)=
    \begin{cases}
      (k-m+1,\,k+1,\,k-(m'-m)+1) & \text{if }m<m',\\
      (k-m+1,\,k+1,\,k-m+m'+1) & \text{if }m>m'.
    \end{cases}
  \]
By Lemma~\ref{2-4-3lemma_1}, $\lvert B-S\rvert\geq 3$.
\end{description}
In any case, we conclude $\lvert B\rvert\geq 6$, which is the required contradiction.

\subsection{$n\equiv 0 \pmod 8$}
We write $n=8k+8$, and prove $\dim(G)\geq 6$. Similar to the previous subsection we
assume that there is a resolving set $B=\{w_1,w_2,w_3,w_4,w_5\}$ such that $w_1<w_2<w_3<w_4<w_5$.
Without loss of generality $w_1=0$ and $w_2,w_3\in\{3,4,\dotsc,4k+3\}$, and we make a case
distinction with respect to the set $S=\{0,\,w_2,\,w_3\}$.
\begin{description}
\item[Case 1] $w_2\equiv w_3\pmod 4$, that is $S=\{0,\,4m+i,\,4m'+i\}$ with $i\in\{0,1,2,3\}$ and
  $1\leq m<m'\leq k$. The vertex set $A=\{4k+i+1,\,4k+i+2,\,4k+i+3,\,4k+i+4\}$ is an $S$-block with
  representing vector $r(A\vert S)=(k+1,\,k-m+1,\,k-m'+1)$. Since $B-S$ resolves $A$,
  Lemma~\ref{k-2resolving} implies $\lvert B-S\rvert\geq 3$.
\item[Case 2] $S=\{0,\,4m,\,4m'+i\}$, $i\in\{1,2,3\}$, $1\leq m<m'\leq k$. Then
  $A=\{4(m'+k)+\ell\,:\,\ell\in\{4,5,6,7\}\}$ is an $S$-block with
  $r(A\vert S)=(k-m'+1,k+m-m'+1,k+1)$, and by Lemma~\ref{k-2resolving}, $\lvert B-S\rvert\geq 3$.
\item[Case 3] $S=\{0,\,4m+1,\,4m'+i\}$, $i\in\{0,2\}$,\ $1 \leq m <m' \leq k$. Let $A=(A_1,\,A_2)$
  with $A_1=\{4k+3,4k+4\}$ and $A_2=\{4k+3+4m'+\ell:\ell \in\{2,3,4\}$. Then $A$ is an $S$-cluster
  with $r(A_1|S)=(k+1,k-m+1,k-m'+1)$ and $r(A_2|S)=(k-m'+1,k-m'+m+1,k+1)$. By
  Lemma~\ref{2-4-3lemma78}, this implies $\lvert B-S\rvert\geq 3$.
\item[Case 4] $S=\{0,4m+1,4m'+3\}$, $1\leq m<m' \leq k$. In this case
  $A=\{4m+4k+j\,:\,j\in\{4,5,6,7\}$ is an $S$-block with $r(A|S)=(k-m+1,\,k+1,\,k-m'+m+1)$. Hence, by
  Lemma~\ref{k-2resolving}, $|B-S|\geq 3$.
\item[Case 5] $S=\{0,\,4m+2,\,4m'\}$, $1 \leq m <m'\leq k$. Let $A=(A_1,\,A_2)$ with
  $A_1=\{4k+3,4k+4\}$ and $A_2=\{4k+3+4m+\ell\,:\,\ell \in \{2,3,4\}\}$. Then $A$ is an $S$-cluster
  with $r(A_1|S)=(k+1,\,k-m+1,\,k-m'+1)$ and $r(A_2|S)=(k-m+1,\,k+1,\,k+m-m'+2)$. By
  Lemma~\ref{2-4-3lemma78}, this implies $|B-S|\geq 3$.
\item[Case 6] $S=\{0,\,4m+2,\,4m'+1\}$, $1 \leq m <m'\leq k$. Let $A=(A_1,\,A_2)$ with
  $A_1=\{4k+3,\,4k+4,\,4k+5\}$ and $A_2=\{4k+5+4m+j\,:\,j \in\{1,2\}\}$. Then $A$ is an $S$-cluster
  as $r(A_1|S)=(k+1,\,k-m+1,\,k-m'+1)$ and $r(A_2|S)=(k-m+1,\,k+1,\,k+m-m'+2)$. Hence, by
  Lemma~\ref{2-4-3lemma78}, $\lvert B-S \rvert \geq 3$.
\item[Case 7] $S=\{0,\,4m+2,\,4m'+3\}$, $1 \leq m<m'\leq k$. In this case
  $A=\{4k+4m+3+j\,:\,j \in\{1,2,3,4\}\}$ is an $S$-block with
  $r(A|S)=(k-m+1,\,k+1,\,k+m-m'+1)$. Hence, by Lemma~\ref{k-2resolving}, $|B-S|\geq 3$.
\item[Case 8] $S=\{0,\,4m+3,\,4m'+i\}$, $i \in \{0,1\}$, $0 \leq m<m'\leq k$ (Note that in this case
  $m$ can be 0 as well). Then $A=(A_1,\, A_2)$ with $A_1=\{4k+2,\,4k+3\}$ and
  $A_2=\{4k+4m'+3+j\,:\,j\in\{1,2,3\} \}$ is an $S$-cluster with $r(A_1|S)=(k+1,\,k-m,\,k-m'+1)$. Hence,
  by Lemma~\ref{2-4-3lemma78}, $|B-S|\geq 3$.
\item[Case 9] $S=\{0,\,4m+3,\,4m'+2\}$, $0 \leq m<m'\leq k$. Let $A=(A_0,\,A_1,\dotsc,
  A_k,\,B_1,\,B_2,\,B_3)$ where
  \begin{align*}
A_0 &= \{4k+4,4k+5,4k+6\},\qquad A_i=\{4k+4+4i,a+5+4i\}\text{ for }1 \leq i \leq k,\\
B_1 &= \{8k+7,1\},\qquad B_{2}=\{2,4\},\qquad B_{3}=\{4m'+1,4m'+3\}.    
  \end{align*}
  Then $A$ is an $S$-cluster with the following representations:
  \begin{align*}
    r(A_0\,\vert\,S) &= (k+1,\,k-m+1,\,k-m'+1),\\
    r(A_i\,\vert\,S) &=
               \begin{cases}
                 (k-i+1,\,k-m+i+1,\,k-m'+i+1)& \text{if } 1 \leq i \leq m,\\
                 (k-i+1,\,k+m-i+2,\,k+m'-i+2)& \text{if } m< i,
               \end{cases}\\
    r(B_1\,\vert\,S) &= (1,m+1,m'+1),\\
    r(B_2\,\vert\,S) &= (1,m+1,m'),\\
    r(B_3\,\vert\,S) &= (m'+1,m'-m,1).
  \end{align*}
 Hence, by Lemma~\ref{8lemmaAkBk}, $|B-S|\geq 3$.
\end{description}

\subsection{$n \equiv 7 \pmod 8$}
We write $n=8k+7$ and prove $\dim(G)\geq 6$. Similar to the previous subsection we
assume that there is a resolving set $B=\{w_1,w_2,w_3,w_4,w_5\}$ such that $w_1<w_2<w_3<w_4<w_5$.
Without loss of generality $w_1=0$ and $w_2,w_3\in\{2,3\dotsc,4k+3\}$, and we make a case
distinction with respect to the set $S=\{0,w_2,w_3\}$.
\begin{description}
\item[Case 1] $w_2\equiv w_3\pmod 4$, that is $S=\{0,\,4m+i,\,4m'+i\}$ with $i\in\{0,1,2\}$ and
  $1\leq m<m'\leq k$. Then $A=\{4k+i+1,\,4k+i+2,\,4k+i+3,\,4k+i+4\}$ is an $S$-block with
  representing vector $r(A\vert S)=(k+1,k-m+1,k-m'+1)$. Since $B-S$ resolves $A$,
  Lemma~\ref{k-2resolving} implies $\lvert B-S\rvert\geq 3$.
\item[Case 2] $S=\{0,\,4m,\,4m'+i\}$, $i\in\{1,2\}$, $1\leq m<m'\leq k$. Then
  $A=\{4(m'+k)+j\,:\,j\in\{3,4,5,6\}\}$ is an $S$-block with $r(A\vert S)=(k-m'+1,\,k+m-m'+1,k+1)$,
  and by Lemma~\ref{k-2resolving}, $\lvert B-S\rvert\geq 3$.
\item[Case 3] $S=\{0,\,4m,\,4m'+3\}$, for $1 \leq m \leq m'\leq k-1$. Let $A=(A_1,\,A_2)$ with
  $A_1=\{4k+2,\,4k+3\}$ and $A_2=\{4k+2+4m+j\,:\,j\in\{2,3,4\}\}$. Then $A$ is an $S$-cluster with
  $r(A_1|S)=(k+1,\,k-m+1,\,k-m')$ and $r(A_2|S)=(k-m+1,\,k+1,\,k+m-m'+1)$. Hence, by
  Lemma~\ref{2-4-3lemma78}, $|B-S|\geq 3$.
\item[Case 4] $S=\{0,\,4m+1,\,4m'\}$, for $1\leq m<m'\leq k$. Let $A=(A_1,\,A_2)$ with $A_1=\{4k+2,\,4k+3\}$
  and $A_2=\{4k+2+4m'+j\,:\,j \in\{2,3,4\}\}$. Then $A$ is an $S$-cluster with
  $r(A_1|S)=(k+1,\,k-m+1,\,k-m'+1)$ and $r(A_2|S)=(k-m'+1,\,k+m-m'+1,\,k+1)$. Hence, by
  Lemma~\ref{2-4-3lemma78}, $|B-S|\geq 3$.
\item[Case 5] $S=\{0,\,4m+1,\,4m'+3\}$, for $1\leq m<m'\leq k-1$. Let $A=(A_1,\,A_2)$ with
  $A_1=\{4k+2,\,4k+3\}$ and $A_2=\{4k+2+4m'+j\,:\,j \in\{2,3,4\}\}$. Then $A$ is an $S$-cluster with
  $r(A_1|S)=(k+1,\,k-m+1,\,k-m'+1)$ and $r(A_2|S)=(k-m'+1,\,k+m-m'+1,\,k+1)$. Hence, by
  Lemma~\ref{2-4-3lemma78}, $|B-S|\geq 3$.
\item[Case 6] $S=\{0,\,4m+1,\,4m'+2\}$, for $1\leq m <m' \leq k$. Then
  $A=\{4k+4m+2+j\,:\,j \in \{1,2,3,4\}\}$ is an $S$-block with representation
  $r(A|S)=(k-m+1,\,k+1,\,k+m-m'+1)$. By Lemma~\ref{k-2resolving}, $|B-S|\geq 3$.
\item[Case 7] $S=\{0,\,4m+2,\,4m'\}$, for $0\leq m <m' \leq k$. Let
  $A=(A_0,\,B_0,\,A_1,\,B_1, \dotsc, A_k,\,B_k,\,A_{k+1})$ with $A_i=\{4k+1+4i,\,4k+2+4i\}$ and
  $B_i=\{4k+3+4i,\,4k+4+4i\}$. Then $A$ is an $S$-cluster with the following representations:
  \begin{align*}
    r(A_0|S) &= (k+1,\,k-m,\,k-m'+1),\\
    r(A_i|S) &=
               \begin{cases}
                 (k-i+2,\,k+i-m,\,k+i-m'+1) & \text{for } 1 \leq i \leq m+1,\\
                 (k-i+2,\,k+m-i+2,\,k+i-m'+1) & \text{for } m+1 < i \leq m',\\
                 (k-i+2,\,k+m-i+2,\,k+m'-i+2)& \text{for } m'<i \leq k+1,
               \end{cases}\\
    r(B_i|S) &=
               \begin{cases}
                 (k-i+1,\,k+i-m+1,\,k+i-m'+1)& \text{for }0 \leq i \leq m, \\
                 (k-i+1,\,k+m-i+2,\,k+i-m'+1)& \text{for }m < i \leq m',\\
                 (k-i+1,\,k+m-i+2,\,k+m'-i+1)& \text{for }m'<i\leq k.
               \end{cases}
  \end{align*}
 Hence, by Lemma \ref{7lemmaAk}, $|B-S|\geq 3$.
\item[Case 8] $S=\{0,\,4m+2,\,4m'+1\}$, $0\leq m <m' \leq k$. Let $A=(A_0,\,A_1,\dotsc A_k,\,B_1,\,B_2)$,
  where
  \begin{align*}
    A_i &= \{3+4i,\,4+4i\},\text{ for }0 \leq i \leq k-1, & A_k&=\{3+4k,\,4+4k,\,5+4k\},\\
    B_{1} &= \{3+4(k+m'),\,4+4(k+m')\}, & B_{2} &=\{5+4(k+m'),\,6+4(k+m')\}.   
  \end{align*}
  Then $A$ is an $S$-cluster with with the following representations.
  \begin{align*}
    r(A_i|S) &=
               \begin{cases}
                 (i+1,\,m-i,\,m'-i), & \text{for }  0 \leq i <m,\\
                 (i+1,\,i-m+1,\,m'-i), & \text{for } m \leq i <m',\\
                 (i+1,\,i-m+1,\,i-m'+1),& \text{for } m' \leq i \leq k,
               \end{cases}\\
    r(B_1|S) & =(k-m'+1,\,k+m-m'+2,\,k+1),\\
    r(B_2|S) &= (k-m'+1,\,k+m-m'+1,\,k+1).    
  \end{align*}
By Lemma~\ref{7lemmaAkBk}, $|B-S|\geq 3$.
\item[Case 9] $S=\{0,\,4m+2,\,4m'+3\}$, for $0\leq m <m' \leq k-1$. Let $A=(A_1,\,A_2)$, where
  $A_1=\{4k+4m+j\,:\,j\in\{4,5,6\}\}$ and $A_2=\{4k+4m'+j\,:\,j \in \{7,8\}\}$. Then $A$ is an
  $S$-cluster with the representations $r(A_1|S)=(k-m+1,\,k+1,\,k+m-m'+1)$ and
  $r(A_2|S)=(k-m',\,k+m-m'+1,\,k+1)$. Hence, by Lemma~\ref{2-4-3lemma78}, $|B-S|\geq 3$.
\item[Case 10] $S=\{0,\,4m+3,\,4m'\}$, for $0 \leq m <m'\leq k$. Let $A=(A_1,\,A_2)$, where
  $A_1=\{4k+5,\,4k+6\}$ and $A_2=\{4k+4m'+j\,:\,j \in\{3,4,5\}\}$. Then $A$ is an $S$-cluster with
  representations $r(A_1|S)=(k+1,\,k-m+1,\,k-m'+2)$ and $r(A_2|S)=(k-m'+1,\,k+m-m'+2,\,k+1)$. By
  Lemma~\ref{2-4-3lemma78}, $|B-S|\geq 3$.
\item[Case 11] $S=\{0,\,4m+3,\,4m'+1\}$,for $0\leq m<m'\leq k$. Let $A=(A_1,A_2,A_3,A_4)$, where
  $A_1=\{1,2\}$, $A_2=\{4k+2,\,4k+3\}$,
  $A_3=\{4k+4,\,4k+5\},A_4=\{4k+4m+j\,:\,j\in\{7,8,9\}\}$. Then $A$ is an $S$-cluster with
  representations $r(A_1|S)=(1,\,m+1,\,m')$, $r(A_2|S)=(k+1,\,k-m,\,k-m'+1)$,
  $r(A_3|S)=(k+1,\,k-m+1,\,k-m'+1)$, $r(A_4|S)=(k-m,\,k+1,\,k+m-m'+2)$. Hence, by
  Lemma~\ref{2-22-3lemma7}, $|B-S|\geq 3$.
\item[Case 12] $S=\{0,\,4m+3,\,4m'+2\}$, for $0 \leq m <m'\leq k$. Let $A=(A_1,\,A_2)$, where
  $A_1=\{4k+5,\,4k+6\}$ and $A_2=\{4k+4m+j\,:\,j \in\{7,8,9\}\}$. Then $A$ is an $S$-cluster with
  representations $r(A_1|S)=(k+1,\,k-m+1,\,k-m'+1)$ and $r(A_2|S)=(k-m,\,k+1,\,k+m-m'+2)$. By Lemma
  \ref{2-4-3lemma78}, this implies $|B-S|\geq 3$.
\item[Case 13] $S=\{0,\,4m+3,\,4m'+3\}$, for $0 \leq m <m'\leq k$. Let $A=(A_1,\,A_2)$, where
  $A_1=\{4k+5,\,4k+6\}$ and $A_2=\{4k+4m+j\,:\,j \in\{7,8,9\}\}$. Then $A$ is an $S$-cluster with
  representations $r(A_1|S)=(k+1,\,k-m+1,\,k-m'+1)$ and $r(A_2|S)=(k-m',\,k+m-m'+1,\,k+1)$. By
  Lemma~\ref{2-4-3lemma78}, this implies $|B-S|\geq 3$.
\end{description}

\subsection{$n \equiv 5 \pmod 8$}
We write $n=8k+5$ and prove $\dim(G)\geq 5$. Assume that there is a resolving set
$B=\{w_1,w_2,w_3,w_4\}$ such that $w_1<w_2<w_3<w_4$. Without loss of generality $w_1=0$ and
$w_2\in\{1,2,\dotsc,4k+2\}$, and we make a case distinction with respect to the set $S=\{0,w_2\}$.
\begin{description}
\item[Case 1] $S=\{0,\,4m\}$, for $1 \leq m \leq k$. Let $A=\{8k+4,\,8k+3,\,8k+2,\,8k+1\}$. Then $A$
  is an $S$-block with the representation $r(A|S)=(1,\,m+1)$. By Lemma~\ref{k-2resolving}, this
  implies $|B-S|\geq 3$.
\item[Case 2] $S=\{0,4m+1\}$, for $1 \leq m \leq k$. Let $A=\{1, 2, 3, 4\}$. Then $A$ is an
  $S$-block with the representation $r(A|S)=(1,m)$. This implies that by Lemma~\ref{k-2resolving},
  $|B-S|\geq 3$.
\item[Case 3] $S=\{0,4m+2\}$, for $1 \leq m \leq k$. Let $A=(A_1,A_2)$, where $A_1=\{8k+1,8k+2\}$
  and $A_2=\{2,3,4\}$. Then $A$ is an $S$-cluster with representations $r(A_1|S)=(1,m+1)$, for
  $1 \leq m <k$, $r(A_1|S)=(1,k)$ for $m=k$ and $r(A_2|S)=(1,m)$. This implies that $|B-S|\geq 3$ by
  Lemma~\ref{2-4-3lemma_r56}.
\item[Case 4] $S=\{0,4m+3\}$, for $1 \leq m < k$. Let $A=(A_1,A_2)$, where $A_1=\{8k+1,8k+2,8k+3\}$
  and $A_2=\{3,4\}$. Then $A$ is an $S$-cluster with representations $r(A_1|S)=(1,m+2)$ and
  $r(A_2|S)=(1,m)$. Hence, by Lemma~\ref{2-4-3lemma_r56}, we have $|B-S|\geq 3$.
\item[Case 5] $S=\{0,1\}$. Consider $A=\{8k+2,8k+3,8k+4,2,3,4\}$. Clearly $A$ is an $S$-block with
  the representation $r(A|S)=(1,1)$. Hence, by Lemma~\ref{r5_lemma_01}, $|B-S|\geq 3$.
\item[Case 6] $S=\{0,2\}$. Let $A=(A_1,A_2)$, where $A_1=\{8k+1,8k+2\}$ and
  $A_2=\{8k+3,8k+4,1,3,4\}$. Then $A$ is an $S$-cluster with representations $r(A_1|S)=(1,2)$ and
  $r(A_2|S)=(1,1)$. This implies that by Lemma~\ref{r5_lemma_02}, $|B-S|\geq3$.
\item[Case 7] $S=\{0,3\}$. Let $A=(A_1,A_2)$, where $A_1=\{8k+1,8k+2,8k+3\}$, and
  $A_2=\{8k+4,1,2,4\}$. Then $A$ is an $S$-cluster with representations as $r(A_1|S)=(1,2)$ and
  $r(A_2|S)=(1,1)$. Hence by Lemma~\ref{r5lemma_03}, $|B-S|\geq 3$.
\end{description}
The above cases are summarized in the table below. The first column has the different choices for
$w_2$, the second column has the $S$-cluster generated by $S=\{0,w_2\}$ and the last column gives
the Lemma which gives the contradiction.
\[
	\begin{tabular}{*{3}{c}}
		\toprule
		         $w_2$          &              $S-cluster$              &           Lemma            \\ \midrule
		 $4m$, $1 \leq m \leq k$  &      $\{8k+4,\,8k+3,\,8k+2,\,8k+1\}$       &   \ref{k-2resolving}  \\
		$4m+1$, $1 \leq m \leq k$ &           $\{1,\, 2,\, 3,\, 4\}$            &   \ref{k-2resolving}  \\
		$4m+2$, $1 \leq m \leq k$ &     $(\{8k+1,\,8k+2\},\{2,\,3,\,4\})$     & \ref{2-4-3lemma_r56} \\
		$4m+3$, $1 \leq m \leq k$ &    $(\{8k+1,\,8k+2,\,8k+3\},\{3,\,4\})$    & \ref{2-4-3lemma_r56} \\
		           1            &      $\{8k+2,\,8k+3,\,8k+4,\,2,\,3,\,4\}$       &   \ref{r5_lemma_01}   \\
		           2            & $(\{8k+1,\,8k+2\},\,\{8k+3,\,8k+4,\,1,\,3,\,4\})$ &   \ref{r5_lemma_02}   \\
		           3            & $(\{8k+1,\,8k+2,\,8k+3\},\{8k+4,\,1,\,2,\,4\})$ &   \ref{r5lemma_03}   \\ \bottomrule
	\end{tabular}
\]

\subsection{$n \equiv 2\pmod 8$}
We write $n=8k+2$ and prove that this implies $\dim(G)=5$. Suppose $B=\{w_1,w_2,w_3,w_4\}$ be a
resolving set of $G$ with $w_1 < w_2 <w_3 <w_4$. Without loss of generality, $w_1=0$ and
$w_2 \in \{1,2, \dotsc,4k\}$. We make a case distinction with respect to the possibilities for the
set $S=\{0,\,w_2\}$.
\begin{description}
\item[Case 1] $S=\{0,\,4m\}$, for $1 \leq m \leq k-1$. Then $A=\{8k-2,\,8k-1,\,8k,\,8k+1\}$ is an
  $S$-block with the representation $r(A|S)=(1,\,m+1)$. By Lemma~\ref{k-2resolving}, this implies
  $|B-S|\geq 3$.
\item[Case 2] $S=\{0,\,4m+1\}$, for $1 \leq m \leq k-1$. Then $A=\{1,2,3,4\}$ is an $S$-block
  with representation $r(A|S)=(1,\,m)$. By Lemma~\ref{k-2resolving}, this implies $|B-S|\geq 3$.
\item[Case 3] $S=\{0,\,4m+2\}$, for $1 \leq m \leq k-2$. Then $A=(A_1,\,A_2)$ with
  $A_1=\{8k-2,8k-1\}$ and $A_2=\{2,3,4\}$ is an $S$-cluster with representations $r(A_1|S)=(1,m+1)$
  and $r(A_2|S)=(1,m)$. This implies that $|B-S|\geq 3$ by Lemma~\ref{2-4-3lemma_r56}.
\item[Case 4] $S=\{0,\,4m+3\}$, for $1 \leq m < k-2$. Then $A=(A_1,\,A_2)$ with
  $A_1=\{8k-2,\,8k-1,\,8k\}$ and $A_2=\{3,\,4\}$ is an $S$-cluster with representations
  $r(A_1|S)=(1,\,m+2)$ and $r(A_2|S)=(1,\,m)$. Hence, by Lemma~\ref{2-4-3lemma_r56}, $|B-S|\geq 3$.
\item[Case 5] $S=\{0,\,1\}$. Then $A=\{8k-1,\,8k,\,8k+1,\,2,\,3,\,4\}$ is an $S$-block with representation $r(A|S)=(1,1)$. By Lemma~\ref{r5_lemma_01}, $|B-S|\geq 3$.
\item[Case 6] $S=\{0,\,2\}$. Then $A=(A_1,\,A_2)$ with $A_1=\{8k-2,\,8k-1\}$ and
  $A_2=\{8k,\,8k+1,\,1,\,3,\,4\}$ is an $S$-cluster with representations $r(A_1|S)=(1,\,2)$
  and $r(A_2|S)=(1,\,1)$. By Lemma~\ref{r5_lemma_02}, $|B-S|\geq3$.
\item[Case 7] $S=\{0,\,3\}$. Then $A=(A_1,\,A_2)$ with $A_1=\{8k-2,\,8k-1,\,8k\}$ and
  $A_2=\{8k+1,\,1,\,2,\,4\}$ is an $S$-cluster with representations $r(A_1|S)=(1,\,2)$ and
  $r(A_2|S)=(1,\,1)$. Hence, by Lemma~\ref{r5lemma_03}, $|B-S|\geq 3$.
\item[Case 8] $S=\{0,\,4k\}$. Then $A=\{8k-2,\,8k-1,\,8k,\,1,\,2,\,3\}$ is an $S$-cluster with
  $r(A|S)=(1,\,k)$. Hence, by Lemma~\ref{r5_lemma_01}, $|B-S|\geq 3$.
\item[Case 9] $S=\{0,\,4k-1\}$. Then $A=(A_1,\,A_2)$ with $A_1=\{4,\,3\}$ and
  $A_2=\{2,\,1,\,8k+1,\,8k-1,\,8k-2\}$ is an $S$-cluster with representations $r(A_1|S)=(1,\,k-1)$ and
  $r(A_2|S)=(1,\,k)$. Hence, by Lemma~\ref{r5_lemma_02}, $|B-S|\geq 3$.
\item[Case 10] $S=\{0,\,4k-2\}$. Then $A=(A_1,\,A_2)$ with $A_1=\{4,\,3,\,2\}$ and
  $A_2=\{1,\,8k+1,\,8k,\,8k-2\}$ is an $S$-cluster with representations $r(A_1|S)=(1,\,k-1)$ and
  $r(A_2|S)=(1,\,k)$. Hence, by Lemma~\ref{r5lemma_03}, $|B-S|\geq 3$.
\end{description}
The above cases are summarised in the table below. 
\[
\begin{tabular}{*{3}{c}}
	\toprule
	          $w_2$           &              $S$-cluster              &          Lemma           \\ \midrule
	 $4m,\ 1 \leq m \leq k-1$  &        $\{-1,\, -2,\, -3,\, -4\}$         &  \ref{k-2resolving} \\
	$4m+1,\ 1 \leq m \leq k-1$ &          $\{1,\, 2,\, 3,\, 4\}$           &  \ref{k-2resolving} \\
	$4m+2,\ 1 \leq m \leq k-2$ &      $\{-4,\, -3\},\, \{2,\, 3,\, 4\}$      &   \ref{2-4-3lemma_r56}     \\
	$4m+3,\ 1 \leq m \leq k-2$ &     $\{-4,\, -3,\, -2\},\, \{3,\, 4\}$      &   \ref{2-4-3lemma_r56}     \\
	            1             &     $\{-3,\, -2,\, -1,\, 2,\, 3,\, 4  \}$     &    \ref{r5_lemma_01}    \\
	            2             & $\{ -4,\, -3 \}, \{-2,\, -1,\, 1,\, 3,\, 4\}$ &    \ref{r5_lemma_02}    \\
	            3             &  $\{ -4,\, -3,\, -2 \}\{-1,\, 1,\, 2,\, 4\}$  &    \ref{r5lemma_03}    \\
	          $4k$            &     $\{-4 ,\, -3,\, -2,\, 1,\, 2,\, 3\}$      &   \ref{r5_lemma_01}    \\
	         $4k-1$           &  $\{4,\, 3\},\, \{2,\, 1,\, -1,\, -3,\, -4\}$   &    \ref{r5_lemma_02}    \\
	         $4k-2$           &  $\{4,\, 3,\, 2\},\,\{1,\, -1,\, -2,\, -4 \}$   &    \ref{r5lemma_03}    \\ 
 \bottomrule
\end{tabular}
\]

\subsection{$n \equiv 3 \pmod 8$}
We write $n=8k+3$ with $k\geq 3$ and prove $\dim(G)\geq 5$. Suppose $B=\{w_1,w_2,w_3,w_4\}$ is a
resolving set of $G$ with $w_1 < w_2 <w_3 <w_4$. Without loss of generality, $w_1=0$ and
$w_2,w_3\in\{1,2,\dotsc,4k+1\}$. We make a case distinction with respect to the possibilities for
the set $S$.
\begin{description}
\item[Case 1] $S=\{0,\,4m\}$, $1 \leq m \leq k-1$. Then $A=\{8k+2,\,8k+1,\,8k,\,8k-1\}$ is an
  $S$-block with the representation $r(A|S)=(1,\,m+1)$. Hence, by Lemma~\ref{k-2resolving},
  $|B-S|\geq 3$.
\item[Case 2] $S=\{0,\,4m+1\}$, $1 \leq m \leq k$. Then $A=\{1,\,2,\,3,4\}$ is an $S$-block with
  the representation $r(A|S)=(1,\,m)$. Hence, by Lemma~\ref{k-2resolving}, $|B-S|\geq 3$.
\item[Case 3] $S=\{0,\,4m+2\}$, $1 \leq m \leq k$. Then $A=\{4k+3,\,4k+4,\,4k+5,\,4k+6\}$ is an
  $S$-block with the representation $r(A|S)=(k,\,k-m+1)$. Hence, by Lemma~\ref{k-2resolving},
  $|B-S|\geq 3$.
\item[Case 4] $S=\{0,\,4m+3,\,w_3\}$, $1 \leq m \leq k-2$ and $w_3=4m'+i$ $i\in\{0,1,2,3\}$ and
  $m<m'$. Now depending on the choice of $w_3$, we have the following possibilities as given in the
  table below. The first column gives the possibilities of $w_3$, the second column has the set $A$
  which is an $S$-cluster, the third column gives the representation $r(A|S)$, and the last column
  contains the lemma by which we have $|B-S|\geq 2$.
  \[\begin{tabular}{*{4}{c}}
      \toprule
      $w_3$&$A$ ( $S$-cluster)&$r(A|S)$& Lemma \\
      \midrule
      $4m',\ m < m' \leq k$& $\{4k-1,\,4k-2,\,4k-3\} $&$(k,\,k-m-1,\,k-m')$& \ref{k-2resolving}  \\ 
      $4m'+1,\  m< m' \leq k$& $\{4m+1,\,4m+2,\,4m+4\}$ &$(m+1,\,1,\,m'-m)$& \ref{k-2resolving} \\ 
      $4m'+2,\ m< m' \leq k-1$ & $\{4k+4,\,4k+5,\,4k+6\}$ &$(k,\,k-m+1,\,k-m'+1)$& \ref{k-2resolving}  \\ 
      $4m'+3,\ m< m' \leq k-1$& $\{4k+4,\,4k+5,\,4k+6\}$ &$(k,\,k-m+1,\,k-m'+1)$& \ref{k-2resolving} \\ 
      \bottomrule
    \end{tabular}\]
\item[Case 5] $S=\{0,\,1,\,w_3\}$, where $w_3 \in \{2,3,\dotsc,4k+1\}$.
  \[\begin{tabular}{*{4}{c}}
      \toprule
      $w_3$& $A$ ($S$-cluster)&$r(A|S)$& Lemma \\ 
      \midrule
      $4m',\ 1\leq m' \leq k-1$& $\{8k+2,\,8k+1,\,8k\} $&$(1,\,1,\,m'+1)$& \ref{k-2resolving} \\ 
      $4m'+1,\ 1\leq m' \leq k$& $\{2,\,3,\,4\}$ &$(1,\,1,\,m')$& \ref{k-2resolving} \\ 
      $4m'+2,\ 1\leq m' \leq k-1$ & $\{4k+4,\,4k+5,\,4k+6\}$ &$(k,\,k,\,k-m'+1)$& \ref{k-2resolving}  \\ 
      $4m'+3,\ 1\leq m' \leq k-1$& $\{4k+4,\,4k+5,\,4k+6\}$ &$(k,\,k,\,k-m'+1)$& \ref{k-2resolving}  \\ 
      2 & $\{8k+1,\,8k+2,\,3,\,4\}$ &(1,\,1,\,1)& \ref{mod3lemma2-3-2} \\ 
      3 & $(\{8k,\,8k+1\},\{8k+2,\,2,\,4\},\{6,\,7\})$&((1,\,1,\,2),\,(1,\,1,\,1),\,(2,\,2,\,1))& \ref{mod3lemma2-7-2a} \\ 
      $4k$ & $(\{4k-2,\,4k-1\},\,\{4k+5,\,4k+6\})$ &$((k,\,k,\,1),\,(k,\,k,\,2))$ & \ref{mod3lemma2-5-2} \\ 
      \bottomrule
    \end{tabular}\]
  The assumption $k\geq 3$ is used for $S=\{0,\,1,\,3\}$: for $k=1$ the set $\{6,\,7\}$ is not an
  $S$-block because $d(0,6)\neq d(0,7)$.
\item[Case 6] $S=\{0,\,2,\,w_3\}$, where $w_3 \in \{3,4,\dotsc,4k+1\}$. 
  \[
    \begin{mytabularwrap}
      \begin{tabular}{*{4}{c}}
	\toprule
	$w_3$&$A$ ($S$-cluster)&$r(A|S)$& Lemma \\ 
	\midrule
	$4m',\ 1\leq m' \leq k$& $(\{4k+1,\,4k+2\},\,\{4k+3,\,4k+4\},\,\{4k+5,\,4k+6\}) $& $\rho_1$ & \ref{mod3lemma222}  \\ 
	$4m'+1,\ 1\leq m' \leq k$& $\{1,\,3,\,4\}$ &$(1,\,1,\,m')$& \ref{k-2resolving} \\ 
	$4m'+2,\ 1\leq m' \leq k-1$ & $\{4k+1,\,4k+2\},\,\{4k+3,\,4k+4\},\,\{4k+5,\,4k+6\}$ & $\rho_2$ & \ref{mod3lemma222}  \\ 
	$4m'+3,\ 1\leq m' \leq k-2$& $\{8k-1,\,8k\},\,\{3,\,4\}$ &$((1,\,2,\,m'+2),\,(1,\,1,\,m'))$& \ref{mod3lemma2-5-2}   \\ 
	3 & $(\{8k-1,\,8k\},\,\{8k+2,\,1,\,4\})$ &((1,\,2,\,2),\,(1,\,1,\,1))& \ref{mod3lemma2-7-2b}  \\ 
	$4k-1$ & $(\{4k-2,\,4k-3\},\,\{4k-4,\,4k-5\},\,\{4k-6,\,4k-7\})$&$\rho_3$& \ref{mod3lemma222} \\ 
	\bottomrule
      \end{tabular}
    \end{mytabularwrap}
  \]
  The representations $\rho_1$, $\rho_2$ and $\rho_3$ in the table are
  \begin{align*}
    \rho_1 &= ((k+1,\,k,\,k-m'+1),\,(k,\,k+1,\,k-m'+1),\,(k,\,k,\,k-m'+2)),\\
    \rho_2 &= ((k+1,\,k,\,k-m'),\,(k,\,k+1,\,k-m'+1),\,(k,\,k,\,k-m'+1)),\\
    \rho_3 &= ((k,\,k-1,\,1),\,(k-1,\,k-1,\,1),\,(k-1,\,k-2,\,2)).        
  \end{align*}
The assumption $k\geq 3$ is used for $S=\{0,\,1,\,4k-1\}$.
\item[Case 7] $S=\{0,\,3,\,w_3\}$, where $w_3 \in \{4,5, \dotsc, 4k+1\}$. 
  \[
    \begin{mytabularwrap}
      \begin{tabular}{*{4}{c}}
	\toprule
	$w_3$&$A$ ($S$-cluster)&r(A|S)& Lemma \\ 
	\midrule
	$4$& $\{5,\,6,\,7\} $&(2,\,1,\,1)& \ref{k-2resolving}  \\ 
	$4m',\ 2\leq m' \leq k$& $\{5,\,6,\,7\} $&$(2,\,1,\,m'-1)$& \ref{k-2resolving}  \\ 
	$4m'+1,\ 1\leq m' \leq k$& $\{1,\,2,\,4\}$&$(1,\,1,\,m')$ & \ref{k-2resolving} \\ 
	$4m'+2,\ 1\leq m' \leq k-1$ & $(\{4k+1,\,4k+2\},\,\{4k+4,\,4k+5\})$ &$((k+1,\,k,\,k-m'),\,(k,\,k+1,\,k-m'+1))$& \ref{mod3lemma2-1-2}  \\ 
	$4m'+3,\ 1\leq m' \leq k-1$& $\{-1,\,1,\,2\}$ &$(1,\,1,\,m'+1)$& \ref{k-2resolving} \\ 
	\bottomrule
      \end{tabular}
    \end{mytabularwrap}
  \]

\item[Case 8] $S=\{0,\,4k,\,4k+1\}$. Then $A=\{4k-1,\,4k-2,\,4k-3\}$ is an $S$-block with the
  representation $r(A|S)=(k,\,1,\,1)$. Hence, by Lemma~\ref{k-2resolving}, $|B-S|\geq 2$.
\item[Case 9] $S=\{0,\,4k-1,\,w_3\}$, where $w_3\in\{4k,\,4k+1\}$.
\[
\begin{tabular}{*{4}{c}}
\toprule
$w_3$&$A$ ($S$-cluster)&$r(A|S)$& Lemma \\ 
\midrule
$4k$& $(\{8k-3,\,8k-2\},\,\{1,\,2\})$ &$((2,\,k,\,k),\,(1,\,k,\,k))$& \ref{mod3lemma2-5-2}  \\ 
$4k+1$& $\{4k,\,4k-2,\,4k-3\}$ &$(k,\,1,\,1)$& \ref{k-2resolving} \\ 
\bottomrule
\end{tabular}\]	
\end{description}

\section{Upper bounds}\label{sec:upper}
The upper bound for the cases which can not be derived from~\cite{Vetrik2016}, are proved in this section. 

\begin{lemma}
Let $G=C(n,\pm\{1,2,3,4\})$ be a circulant graph with $n=8k+9$. Then $\dim(G)\leq 6$.
\end{lemma}
\begin{proof}
  We show that the set $X=\{0,1,4,7,4k+6,4k+7\}$ is a metric basis for $G$. For any two vertices
  $a,b\in V-X$, we need to show that $d(a,x)\neq d(b,x)$ for some $x\in X$. Writing $a=4m_1+r_1$ and
  $b=4m_2+r_2$ with $m_1,m_2\in \{0,1,\dotsc,2k+1\}$ and $ r_1,r_2 \in \{1,2,3,4\}$, we have the
  following cases.
\begin{description}
\item[Case 1] $m_1,m_2\leq k$. If $m_1 \neq m_2$, then
  $d(0,a)\neq d(0,b)$, as $d(0,a)=m_1+1$ and $d(0,b)=m_2+1$. If $m_1=m_2$, without loss of
  generality $r_1<r_2$. The following list describes how $a$ and $b$ are resolved for each of the
  possible values of $a$.
  \begin{align*}
    a=4m_1+1,\ 0\leq m_1 \leq k &\implies d(1,a)=m_1,\ d(1,b)=m_1+1, \\
    a=4m_1+2,\ 1\leq m_1 \leq k &\implies d(4k+7,a)=k-m_1+2,\ d(4k+7,b)=k-m_1+1, \\
    a=4m_1+3,\ 1\leq m_1 \leq k &\implies d(7,a)=m_1-1,\ d(7,b)=m_1, \\
    a=2 &\implies d(7,a)=2,\ d(7,b)=1, \\      
    a=3 &\implies d(4,a)=1,\ d(4,b)=0.
  \end{align*}
\item[Case 2] $m_1,m_2\geq k+1$. If $m_1 \neq m_2$, then $d(0,a)\neq d(0,b)$, as $d(0,a)=(2k+2)-m_1$
  and $d(0,b)=(2k+2)-m_2$. If $m_1=m_2$, without loss of generality $r_1<r_2$. The following list
  describes how $a$ and $b$ are resolved for each of the possible values of $a$.
  \begin{align*}
      a=4m_1+1,\ k+2\leq m_1 \leq 2k+1 &\implies d(1,a)=2k-m_1+3,\ d(1,b)=2k-m_1+2, \\      
      a=4m_1+2,\ k+1\leq m_1 \leq 2k+1 &\implies d(4k+6,a)=m_1-k-1,\ d(4k+6,b)=m_1-k, \\
    a=4m_1+3,\ k+1\leq m_1 \leq 2k+1 &\implies d(4k+7,a)=m_1-k-1,\ d(4k+7,b)=m_1-k,\\
    a=4k+5 &\implies d(7,a)=k+1,\ d(7,b)=k. 
  \end{align*}
\item[Case 3] $m_1\leq k$ and $m_2\geq k+1$. If $m_1+m_2\neq 2k+1$, then $d(0,a)\neq d(0,b)$, as $d(0,a)=m_1+1$ and $d(0,b)=2k+2-m_2$. If $m_1+m_2=2k+1$, then we have the following possibilities.
  \begin{align*}
    a=4m_1+r_1,\ 1 \leq m_1 \leq k-1 &\implies d(4,a)=m_1,\ d(4,b)=m_1+2, \\    
    a=4k+r_1 &\implies d(4,a)=k, d(4,b)=k+1,\\
    a=r_1 &\implies d(4,a) \in \{1,0\},\ d(4,b)=2.\qedhere
  \end{align*}
\end{description}
\end{proof}
\begin{lemma}
  Let $G=C(n,\pm\{1,2,3,4\})$ with $n=8k+7$. Then $\dim(G)\leq 6$.
\end{lemma}
\begin{proof}
We show that $X=\{0,1,2,3,4,5\}$ is a metric basis for $G$. For any two vertices
$a,b\in V=\{0,\pm1,\pm2, \dotsc, \pm (4k+3)\}$, we need to show that $d(a,x)\neq d(b,x)$ for some
$x\in X$. Let $a=\pm(4m_1+r_1)$ and $b=\pm(4m_2+r_2)$, where
$0 \leq m_1,m_2 \leq k$, $1\leq r_1,r_2 \leq 4$ (if $m_i=k$, then $1 \leq r_i\leq 3$). Note that
$d(0,a)=m_1+1$ and $d(0,b)=m_2+1$, so we may assume $m_1=m_2=m$.
\begin{description}
\item[Case 1] $a=4m+r_1$ and $b=4m+r_2$. Without loss of
  generality $r_1<r_2$, hence $r_1\in\{1,2,3\}$, and $d(r_1,a)=m_1$, $d(r_1,b)=m_1+1$.
\item[Case 2] $a=-(4m+r_1)$ and $b=-(4m+r_2)$. If $m\leq k-1$,
  without loss of generality $r_1<r_2$, hence $r_1\in\{1,2,3\}$, and $d(4-r_1,a)=m_1$,
  $d(4-r_1,b)=m_1+1$. If $m_1=m_2=k$, without loss of generality $r_1<r_2$, hence $r_1\in\{1,2\}$,
  and $d(6-r_1,a)=k+1$, $d(6-r_1,b)=k$.
\item[Case 3] $a=4m+r_1$ and $b=-(4m+r_2)$. If $m=0$ then $d(4,a)\in\{0,1\}$, $d(4,b)=2$. If $1\leq
  m\leq k-1$ then $d(4,a)=m$, $d(4,b)=m+2$. If $m=k$ then $d(4,a)=k$, $d(4,b)=k+1$.\qedhere	
\end{description}
\end{proof}


\providecommand{\bysame}{\leavevmode\hbox to3em{\hrulefill}\thinspace}
\providecommand{\MR}{\relax\ifhmode\unskip\space\fi MR }
\providecommand{\MRhref}[2]{%
  \href{http://www.ams.org/mathscinet-getitem?mr=#1}{#2}
}
\providecommand{\href}[2]{#2}

\end{document}